\documentclass[a4paper,twocolumn,11pt,accepted=2026-01-08]{quantumarticle}
\pdfoutput=1
\usepackage[utf8]{inputenc}
\usepackage[english]{babel}
\usepackage[T1]{fontenc}
\usepackage{amsmath,amssymb,amsthm,bm,amsfonts,bbm} 	
\usepackage{mathtools} 	
\usepackage{hyperref}
\usepackage{soul}
\hypersetup{
    colorlinks=true,
    citecolor=ForestGreen,
    linkcolor=ForestGreen
}

\usepackage{dsfont}
\usepackage[numbers,sort&compress]{natbib}
\usepackage{physics} 	
\usepackage{microtype} 	
\usepackage[dvipsnames,x11names]{xcolor} 	

\usepackage{tikz}
\usepackage{subcaption}
\usepackage{lipsum}

\newtheorem{Theorem}{Theorem}
\newtheorem*{theorem*}{Theorem}
\newtheorem*{Theorem*}{Theorem}
\newtheorem{definition}[Theorem]{Definition}
\newtheorem{proposition}[Theorem]{Proposition}
\newtheorem*{proposition*}{Proposition}
\newtheorem{observation}[Theorem]{Remark}
\newtheorem{Example}[Theorem]{Example}

\newtheorem{lemma}[Theorem]{Lemma}

\newcommand{\blue}[1]{{\color{Violet}#1}}

\newcommand{\new}[1]{{\color{black}#1}} 

\begin{document}

\title{Characterising memory in quantum channel discrimination via constrained separability problems}

\author{Ties-A. Ohst}
\affiliation{Naturwissenschaftlich-Technische Fakult{\"a}t, Universit{\"a}t Siegen, Siegen 57068, Germany}

\author{Shijun Zhang}
\affiliation{Sorbonne Universit\'e, CNRS, LIP6, F-75005 Paris, France}

\author{Hai Chau Nguyen}
\affiliation{Naturwissenschaftlich-Technische Fakult{\"a}t, Universit{\"a}t Siegen, Siegen 57068, Germany}

\author{Martin Pl\'{a}vala}
\affiliation{Naturwissenschaftlich-Technische Fakult{\"a}t, Universit{\"a}t Siegen, Siegen 57068, Germany}
\affiliation{Institut f\"{u}r Theoretische Physik, Leibniz Universit\"{a}t Hannover, Hannover, Germany}

\author{Marco Túlio Quintino}
\affiliation{Sorbonne Universit\'e, CNRS, LIP6, F-75005 Paris, France}

\maketitle

\begin{abstract}
Quantum memories are a crucial precondition in many protocols for processing quantum information. A fundamental problem that illustrates this statement is given by the task of channel discrimination, in which an unknown channel drawn from a known random ensemble should be determined by applying it for a single time. In this paper, we characterise the quality of channel discrimination protocols when the quantum memory, quantified by the auxiliary dimension, is limited. This is achieved by formulating the problem in terms of separable quantum states with additional affine constraints that all of their factors in each separable decomposition obey. \new{Practically, the formulation allows us to introduce computational methods that provide a sequence of rigorous upper and lower bounds on the performance of channel discrimination tasks with restricted memory resources. The practical aspect of our methods is illustrated in single-copy scenarios, and the flexibility of our approach allows us to systematically characterise quantum and classical memories in scenarios with multiple copies, which may involve adaptive or non-adaptive protocols.
In concrete terms, our methods enable the identification of channel discrimination scenarios where classical or quantum memory is required, and to formally understand the role of quantum information transfer in adaptive channel discrimination protocols.
} 
\end{abstract}

{
  \hypersetup{linkcolor=ForestGreen}
  \tableofcontents
}

\section{Introduction}
Memory describes the preservation of information over time and is a fundamental concept for the description of stochastic processes \cite{vanKampen92,Breuer2007book}. 
In quantum theory, memory phenomena are particularly interesting because information is described through quantum states, which can be entangled across different subsystems.

Exactly the interplay between the preservation of quantum information and the potential entanglement of the system with the others plays an important role in the understanding of open quantum dynamics \cite{Li2018physrep,deVega2017RMP,Milz2012PRXQ} and temporal correlations in quantum processes \cite{Taranto2019PRL,Pollock2018PRL,Vieira2022Quantum}. Considering that multiple approaches exist for formalising quantum and classical memory in quantum theory \cite{Kretschmann2005PRA, Taranto2020IJQI, Breuer201RMP, Vieira2024, rosset18, Konig2005, Milz2020PRX}, here we adopt the approach consistent with Refs.~\cite{Giarmatzi2021, Nery2021, Taranto2024characterising}, which defines quantum memory as an auxiliary quantum system to preserve states over time. The size of a quantum memory is then identified with the Hilbert space dimension of this auxiliary system. Analogously, classical memory is quantified by the number of different classical symbols that can be stored in an auxiliary space.

From the practical perspective, quantum memories are known as an important precondition for various information processing tasks such as the superdense coding protocol \cite{PhysRevLett.69.2881}; the latter describes the communication of two classical bits by the implementation of a single qubit operation. This is possible due to a perfect discrimination of four particular qubit channels by only a single call. The channel acts on a part of a maximally entangled state whose counterpart is stored in a quantum memory. The superdense coding protocol can be seen as a particular instance of a protocol to discriminate a given set of quantum channels. Such discrimination problems are of great relevance in the fields of quantum computing and quantum metrology. 

Despite the great potential of quantum memories in numerous applications \cite{Heshami2016}, their physical implementation is difficult and especially the storage of parts of high-dimensional entangled states is prone to noise. \new{In earlier works \cite{PhysRevLett.101.180501, PhysRevA.82.032302, Puzzuoli17, PhysRevLett.102.250501, jenvcova2016conditions}, it has been pointed out that the presence of quantum memory is in general necessary to obtain the optimal performance in a channel discrimination problem.}
This \new{fact} motivates the question of how the amount of quantum memory affects the quality of a quantum information protocol, such as the general channel discrimination task.

\new{
Our main contribution in this article is the introduction of a method to quantitatively determine the optimal success probability in the setting of memory restrictions for arbitrary channel discrimination tasks. This concrete contribution is formalised in Thm.~\ref{prop:symmetric_ext_tester} as an idea which is later generalised in the multi-copy regime in the Sections ~\ref{sec:multi_memory_const} and \ref{sec:class_adp} by characterising different instances of adaptive and non-adaptive protocols.}

The key insight is that the underlying problem can be formulated as a maximisation of a linear functional over separable quantum states whose factors in each separable decomposition obey additional affine linear constraints. Recently, in Ref.~\cite{Berta2021}, a hierarchy of semidefinite programs has been proven to provide a solution as a decreasing and convergent sequence of upper bounds for this class of problems.   
Building on this, we formulate a general method to systematically quantify memory constraints in channel discrimination tasks. As such, it is applicable for general instances of channel discrimination tasks and does not depend on additional structures such as symmetries of the problem. For the computation of lower bounds, we apply the seesaw method, that has been used in the context of memory constrained quantum metrology \cite{kurdzialek2024}. To apply the seesaw algorithm in a rigorous manner, we derive, using the polytope approximation technique~\cite{Cavalcanti2015LHV,Hirsch16,Nguyen19,Nery2021,Ohst24, steinberg2023}, a general bound on the distance from the seesaw value to the true optimum. 

Turning to the case where multiple copies of the unknown channel are accessible, it has been shown that adaptive strategies can improve the quality of the channel discrimination protocol \cite{Harrow2010, Bavaresco21PRL,Bavaresco22JMP, Wilde2020, Salek2022, Li2022}. Also in regard of memory effects, the observed temporal correlations in multi-time processes have been used in Ref.~\cite{Vieira2024} to detect the presence of quantum memory.   

In the case of adaptive strategies for channel discrimination, the role of quantum memory effects is subtle due to the possible presence of classical memory effects. The latter has been investigated recently for multi-time quantum processes \cite{Taranto2024characterising, Bäcker24}. 
Using the framework of constrained separability, here we develop a systematic characterisation of adaptive strategies for channel discrimination in which the amount of quantum and classical memory is limited. \new{Apart from the characterisation of the role of memory, we arrive at new insights into the characterisation of multi-copy strategies. In particular, we show that adaptive strategies, in which the transmission of quantum information between the channel applications is prohibited, have no hierarchical relationship to parallel strategies.}

The manuscript is structured as follows. After fixing the notation in Section \ref{sec:notation}, we start with elaborating the problem of memory constraints in single-copy channel discrimination protocols and the connection to separability problems in Section \ref{sec:single_copy}. We proceed by reviewing multi-copy strategies for channel discrimination in Section \ref{sec:multi_without_const} before we discuss their corresponding memory constraints in Section \ref{sec:multi_memory_const}. In Section \ref{sec:class_adp}, we introduce the class of classically adaptive discrimination schemes. 

\section{Notation and preliminaries}
\label{sec:notation}
We begin by fixing the notation that is used in this paper. \blue{Quantum systems}, labelled by capital letters $\rm{A}$, are associated to finite-dimensional Hilbert spaces $\mathcal{H}_{\rm{A}} \cong \mathbb{C}^{\rm d_A}$ of dimension $d_{\rm A}$. Given a \blue{linear operator} $X$, the system on which it acts on is indicated by a subscript, e.g., the notation $X_{\rm A}$ indicates that $X_{\rm A} \in \mathcal{L}(\mathcal{H}_{\rm A})$. A \blue{state} $\rho_{\rm A}$ on a system $\rm A$ is a positive semidefinite operator, indicated as $\rho_{\rm A} \geq 0$, such that $\Tr(\rho_{\rm A}) = 1$ where $\Tr(\cdot)$ denotes the trace. A \blue{measurement} with $N$ outcomes on a system $\rm{A}$ is a positive operator-valued measure (POVM), i.e., a collection of operators $\{M_{\rm A}^{i}\}_{i=1}^{N}$ such that $M_{\rm A}^{i} \geq 0$ and $\sum_{i=1}^{N} M_{\rm A}^{i} = \mathds{1}_{\rm A}$. \blue{Composite systems} containing subsystems like $\rm{A}$ and $\rm{B}$ are denoted by the string $\rm{AB}$ and their associated Hilbert space is the tensor product $\mathcal{H}_{\rm A} \otimes \mathcal{H}_{\rm B}$. Given a system $\rm{A}$, $\rm{A}^{k}$ denotes the system comprising $k$ copies of $\rm{A}$ and the associated Hilbert space is given $\mathcal{H}_{\rm A}^{\otimes k}$. $\textnormal{Sym}(\rm{A},k)$ denotes the system associated to the \blue{symmetric subspace} of $\mathcal{H}_{\rm A}^{\otimes k}$, i.e., the space spanned by vectors that are invariant under all permutations of the individual copies of $\rm{A}$ \cite{harrow2013}.

A \blue{map} from a system $\rm{A}$ to a system $\rm{B}$, denoted as $\mathcal{D}_{\rm A \rightarrow B}$, is an element from $\mathcal{L}(\mathcal{L}(\mathcal{H}_{\rm A}), \mathcal{L}(\mathcal{H}_{\rm B}))$, i.e., a linear operator from one space of linear operators to another. In case that $\rm{A}=\rm{B}$, we write $\mathcal{D}_{\rm A}$ to abbreviate $\mathcal{D}_{\rm A \rightarrow A}$. Simple but important examples of maps are the transpose $X_{\rm A} \mapsto (X_{\rm A})^T$ and the partial transpose $X_{\rm AB} \mapsto (X_{\rm AB})^{T_{\rm B}}$. A map $\mathcal{D}_{\rm A \rightarrow B}$ is called \blue{positive} if  $\mathcal{D}_{\rm A \rightarrow B}(X_{\rm A})_{\rm B}$ is positive semidefinite for every positive semidefinite $X_{\rm A}$. It is called \blue{completely positive} if $(\textnormal{id}_{\rm E} \otimes \mathcal{D}_{\rm A \rightarrow B})_{\rm EA\rightarrow EB}$ is a positive map for all systems $\rm{E}$. A completely positive map that is trace-preserving is referred to as a \blue{quantum channel}. Simple and important examples for quantum channels are the identity on $\mathcal{L}(\mathcal{H}_{\rm A})$ written as $\textnormal{id}_{\rm A}$ and the partial trace denoted by $\Tr_{\rm A}(\cdot)$. Furthermore, the following channels are important for our discussion. 
\begin{definition}
\label{def:important_channels}
    The map $\mathcal{U}_{\rm A}(X_{\rm A}) = U_{\rm A}X_{\rm A}U_{\rm A}^{\dagger}$ is a quantum channel for every \blue{unitary operator} $U_{\rm A}$. One particular set of unitary operators on a $d$-dimensional system are the \blue{clock-shift operators} $X_d^i Z_d^j$ for $i,j\in \{0,\dots,d-1\}$, 
    where $X_d:=\sum_{l=0}^{d-1} \ketbra{l \oplus 1}{l}, Z_d:=\sum_{l=0}^{d-1} \omega^{l}\ketbra{l}{l}$ and $\omega:= e^{\sqrt{-1} \frac{2\pi}{d}}$. In particular for $d=2$, they correspond to the \blue{Pauli matrices} $\sigma_0 = X_{2}^{0}Z_{2}^{0}$, $\sigma_x = X_{2}^{1}Z_{2}^{0}$, $\sigma_y = i X_{2}^{1}Z_{2}^{1}$ and $\sigma_z = X_{2}^{0}Z_{2}^{1}$. 

   Other important non-unitary qubit channels are the families of \blue{bit-flip channels} $\mathcal{C}_{\textnormal{bf}, p}$ and \blue{amplitude damping channels} $\mathcal{C}_{\textnormal{ad}, p}$ with $p\in [0,1]$ given by 
    \begin{align}
        \mathcal{C}_{\textnormal{bf}, p}(\rho) &:= p \rho +  (1-p) \, \sigma_x \rho \sigma_x \\
        \mathcal{C}_{\textnormal{ad}, p}(\rho) &:= B_0 \rho B_0^{\dagger}+ B_1 \rho B_1^{\dagger}
    \end{align}
    with the operators 
    \begin{equation}
        B_0 = \begin{pmatrix}
            1 & 0 \\ 0 & \sqrt{1-p}
        \end{pmatrix}, \; B_1 = \begin{pmatrix}
            0 & \sqrt{p} \\ 0 & 0
        \end{pmatrix}.
    \end{equation}
\end{definition}
In particular, the unitary transformations given by clock-shift operators will be crucial to demonstrate a strong memory dependence of channel discrimination protocol, see Theorem ~\ref{prop:clock_shift_memory_dependence}.

A \blue{quantum instrument} is a collection of completely-positive maps $\{\mathcal{K}_{\rm A \rightarrow B}^{i}\}_{i=1}^{L}$ that sum up to a quantum channel $\mathcal{K}_{\rm A \rightarrow B}$, i.e., $\mathcal{K}_{\rm A \rightarrow B} = \sum_{i=1}^{L} \mathcal{K}_{\rm A \rightarrow B}^{i}$. Physically, quantum instruments describe the non-deterministic quantum evolutions of states in measurements if the measurement outcome is kept track of.  

The space of maps from $\rm{A}$ to $\rm{B}$ is isomorphic to the space of linear operators on the joint system $\rm{AB}$ that is  $\mathcal{L}(\mathcal{H}_{\rm A} \otimes \mathcal{H}_{\rm B})$.  One particularly important isomorphism is the \blue{Choi–Jamiołkowski isomorphism} \cite{Choi1972, Jamiokowski1972} that transforms a map $\mathcal{D}_{\rm A \rightarrow B}$ into an operator
\begin{align}
    D_{\rm AB} := \sum_{ij}\ketbra{i}{j}_{\rm A} \otimes \mathcal{D}_{\rm A \rightarrow B}(\ketbra{i}{j}_{\rm A})
\end{align}
where $\{\ket{i}\}_{i=0}^{d_{\rm A}-1}$ is an orthonormal basis of $\mathcal{H}_{\rm A}$. We call the operator $D_{\rm AB}$ the \blue{Choi matrix} of the super-operator $\mathcal{D}_{\rm A \rightarrow B}$. An important property of the Choi–Jamiołkowski isomorphism is that completely positive maps are mapped into positive semidefinite Choi-matrices. In particular, the Choi-matrices of quantum channels $\mathcal{C}_{\rm A \rightarrow B}$ can exactly be identified with positive semidefinite operators $C_{\rm AB}$ that satisfy $\Tr_{\rm B} (C_{\rm AB}) = \mathds{1}_{\rm A}$.

A practical tool for composing maps using the Choi matrix is the link product \cite{PhysRevA.80.022339, PhysRevA.77.062112}.
\begin{definition}[\blue{Link product} \cite{PhysRevA.80.022339, PhysRevA.77.062112}]
\label{def:link_product}
    Let $\rm{A}, \rm{E}$ and $\rm{B}$ be quantum systems and let $X_{\rm AE} \in \mathcal{L}(\mathcal{H}_{\rm A} \otimes \mathcal{H}_{\rm E})$ and $Y_{\rm EB} \in \mathcal{L}(\mathcal{H}_{\rm E} \otimes \mathcal{H}_{\rm B})$ be linear operators. The link product $Z_{\rm AB} = X_{\rm AE}*Y_{\rm EB}$ is defined by 
\begin{equation}
    X_{\rm AE}*Y_{\rm EB} := \Tr_{\rm E}\left[((X_{\rm AE})^{T_{\rm E}} \otimes \mathds{1}_{\rm B})(\mathds{1}_{\rm A} \otimes Y_{\rm EB})\right].
\end{equation}
\end{definition}
The link product has a clear operational meaning. The Choi matrix of a composition of two-operators is exactly the link product of the individual Choi matrices. If $\mathcal{X}_{\rm A\rightarrow E}$ and $\mathcal{Y}_{\rm E\rightarrow B}$ are maps, the Choi operator of the composition $(\mathcal{Y} \circ \mathcal{X})_{\rm A \rightarrow B}$ is given
by the link product $X_{\rm AE}*Y_{\rm EB}$ \cite{PhysRevA.80.022339, PhysRevA.77.062112}. Sometimes it is convenient to use a different order of the systems when using the link product. System permutations are then implicitly assumed and suppressed in the notation. 

\section{Single-copy channel discrimination}
\label{sec:single_copy}
We consider the problem of \blue{minimum-error channel discrimination}. One assumes that an unknown quantum channel is drawn from a given ensemble $\{q_i, \mathcal{C}_{\rm I \rightarrow O}^{i}\}_{i=1}^{N}$ consisting of known quantum channels $\mathcal{C}_{\rm I \rightarrow O}^{i}$ and known probabilities $q_i$ of how likely the corresponding channel is drawn. The labels $\rm{I}$ and $\rm{O}$ denote the input and output systems of the channels and the task is to determine the correct channel label $i$ assuming that the unknown channel $\mathcal{C}_{\rm I \rightarrow O}^{i}$ can be applied for a single time.

The general strategy that could be used for this problem exploits the presence of entanglement and a quantum memory and can be described as follows.

First, some quantum state $\rho_{\rm IE}$, potentially entangled with an auxiliary quantum system $\rm{E}$ that acts as quantum memory, is prepared. In the next step, the unknown channel acts on one part of this state resulting in $[\mathcal{C}_{\rm I \rightarrow O}^{i}\otimes \text{id}_{\rm E}] (\rho_{\rm IE})$. Finally, a measurement $\{M_{\rm EO}^{i}\}_{i=1}^{N}$ on the joint system of $\rm{O}$ and $\rm{E}$ with outcomes $i = 1,\dots,N$ is performed and an observed outcome $j$ results in guessing the associated channel $\mathcal{C}_{\rm I \rightarrow O}^{j}$. 

Given a state and measurement as above, the \blue{success probability} $p$ of guessing $j$ correctly using this strategy can be computed as
\begin{align}
        p &= \sum_{i=1}^{N} q_{i} \Tr(M_{\rm OE}^{i} \, [\mathcal{C}_{\rm I \rightarrow O}^{i}\otimes \text{id}_{\rm E}](\rho_{\rm IE})) \label{eq:success_prob_non_linear}\\
    &= \sum_{i=1}^{N} q_{i} \Tr(T_{\rm IO}^{i}  C_{\rm IO}^{i}) \label{eq:single_copy_success}
\end{align} 
where, $C_{\rm IO}^{i}$ is the Choi matrix of the channel $\mathcal{C}_{\rm I \rightarrow O}^{i}$ and $T_{\rm IO}^{i} := \rho_{\rm IE} * {(M_{\rm EO}^{i})}^{T}$, where $*$ is the link product, see Definition \ref{def:link_product}.
The operators $\{T_{\rm IO}^{i}\}_{i=1}^{N}$, which at the same time contain information about the state preparation $\rho_{\rm IE}$ and the measurement $\{M_{\rm EO}^{i}\}_{i=1}^{N}$, will turn out very handy in the following, both from the practical perspective in optimisation problems and from the perspective of generalising the scenario to more copies.
The collection of operators $\{T_{\rm IO}^{i}\}_{i=1}^{N}$ is denoted as a \blue{tester} and testers have been introduced at first in Ref.~\cite{PhysRevA.80.022339, PhysRevA.77.062112}. \new{Conceptually, testers in quantum channel discrimination can be seen as the generalisation of measurements in quantum state discrimination, and will here be used to concisely describe various types of different discrimination protocols. From a broader perspective, testers provide a linearisation of the success probability Eq.~\eqref{eq:success_prob_non_linear} that is otherwise a bi-linear function in the preparation and measurement, respectively.}
\begin{definition}[\blue{Single-copy tester}  \cite{PhysRevA.80.022339, PhysRevA.77.062112}]
\label{def:sc_tester}
    A collection of positive operators $\{T_{\rm IO}^{i}\}_{i=1}^{N} \subset \mathcal{L}(\mathcal{H}_{\rm I} \otimes \mathcal{H}_{\rm O})$ is a single-copy tester if there is a quantum system $\rm{E}$ such that 
    \begin{align}\label{eq:singly_copy_limited_tester}
        &T_{\rm IO}^{i} = \rho_{\rm IE} * (M_{\rm EO}^{i})^T  
    \end{align}
    where $\rho_{\rm IE}$ is a state and $\{M_{\rm EO}^{i}\}_{i=1}^{N}$ is a measurement. The tester can be depicted by the following circuit, where $\mathbb{C}^{\infty}$ indicates that the memory system $\rm{E}$ has no dimension restriction. 
    \begin{center}
\resizebox{0.3\textwidth}{!}{\begin{tikzpicture}[x=1cm,y=1cm]

\filldraw [color=blue!60, fill=blue!5, very thick] 
(0,1.8) arc [start angle=90, end angle=270, x radius=.75, y radius=.9] 
node [pos=0.5,xshift=0.1cm,right,black] {\Large{$\rho$}};
\draw[very thick,color=blue!60](0,0) -- (0,1.8) ;

\draw[->,thick](0,1.4) --(3,1.4) node [pos=.5, above]{$\rm{E} \cong \mathbb{C}^{\infty}$};
\draw[->,thick](0,.4)--(1,.4) node [pos=.5, below]{$\rm{I}$};

\filldraw [color=red!60, fill=red!5, very thick] 
(1,0) rectangle (2,.8) node [pos=0.5,black] {$\mathcal{C}$};

\draw[->,thick](2,.4) --(3,.4) node [pos=.5, below]{$\rm{O}$};

\filldraw [color=blue!60, fill=blue!5, very thick] 
(3,0) arc [start angle=-90, end angle=90, x radius=.75, y radius=.9];
\draw[very thick,color=blue!60]
(3,1.8) -- node [right, black] {$M^{i}$}(3,0) ;

\end{tikzpicture}}
\end{center}
    
\end{definition}
Testers can be seen as generalisation of usual measurements that are described as POVMs. This is due to the fact that testers linearly assign probability distributions to a channel $\mathcal{C}_{\rm I \rightarrow O}$ via $p(i|T,C) = \Tr(T_{\rm IO}^{i} C_{\rm IO})$. In this way, testers play an analogous role to POVMs in the generalised version of Born's rule that assign probabilities to channels instead of states. 
\begin{Example}[Superdense coding]
    To give a concrete example, we examine the superdense coding protocol \cite{PhysRevLett.69.2881}. The input system $\rm{I}$ and output system $\rm{O}$ are both two-dimensional and the channels $\mathcal{C}_{\rm I \rightarrow O}^{j}$ are the Pauli rotations, i.e,  $\mathcal{C}_{\rm I \rightarrow O}^{j}(\cdot) = \sigma_j \cdot \sigma_j$ for $j \in \{0,1,2,3\}$. Its corresponding Choi matrices are the Bell states with $C_{\rm IO}^{0} = 2 \ketbra{\phi^{+}}, C_{\rm IO}^{1} = 2\ketbra{\psi^{+}}, C_{\rm IO}^{2} = 2\ketbra{\psi^{-}}$ and $C_{\rm IO}^{3} = 2\ketbra{\phi^{-}}$. The well known optimal strategy to discriminate Pauli channels involves a qubit memory system $\rm{E}$ and preparing the maximally entangled state $\ket{\phi^{+}}_{\rm IE}$. Then, the unknown Pauli operator acts on the input qubit $\rm{I}$ before a Bell state measurement on the joint system of memory and output $\rm{EO}$ is performed in the end. The corresponding testers $T_{\rm IO}^i$ are given by  
    \begin{align}
        T_{\rm IO}^{0} &= \ketbra{\phi^{+}}_{\rm IE}  * \ketbra{\phi^{+}}_{\rm EO}, \\
        T_{\rm IO}^{1} &= \ketbra{\phi^{+}}_{\rm IE}  * \ketbra{\psi^{+}}_{\rm EO}, \\
        T_{\rm IO}^{2} &= \ketbra{\phi^{+}}_{\rm IE} * \ketbra{\psi^{-}}_{\rm EO}, \\
        T_{\rm IO}^{3} &= \ketbra{\phi^{+}}_{\rm IE} * \ketbra{\phi^{-}}_{\rm EO},
    \end{align}
    and a straightforward calculation reveals $\Tr(T_{\rm IO}^{i} C_{\rm IO}^{j}) = \delta_{ij}$.
\end{Example}

Given a channel ensemble, finding the optimal discrimination strategy corresponds to an optimisation problem over testers. We consider the situation in that the set of available testers is limited by quantum memory quantified by the dimension of the auxiliary system $\rm{E}$. Before that, we start with the discussion of channel discrimination with unlimited quantum memory. 

\subsection{Unlimited memory case} 

 In the case of unlimited quantum memory, one makes the strong assumption that states $\rho_{\rm IE}$ entangled with a system $\rm{E}$ of arbitrarily high dimension can be prepared, and that the coherence is preserved by a perfect identity channel $\textnormal{id}_{\rm E}$ during the channel application on the input system $\rm{I}$. Also, it is assumed that every measurement $\{M_{\rm EO}^{i}\}_{i=1}^{N}$ on the joint system $\rm{EO}$ can be performed. 
 
In this case, the set of possible testers $\{T_{\rm IO}^{i}\}_{i=1}^{N}$ is completely determined in terms of positive semidefinite matrices and affine constraints \cite{PhysRevA.77.062112, PhysRevA.80.022339} leading to the following Theorem.
 \begin{proposition}[Ref. \cite{PhysRevA.77.062112, PhysRevA.80.022339}]
 \label{prop:single_copy_phys_char}
     A collection of positive operators $\{T_{\rm IO}\}_{i=1}^{N}$ is  a single-copy tester if and only if there exists a state $\sigma_{\rm I}$ such that
 \begin{align}\label{eq:single_copy_tester_def}
    &\sum_{i=1}^{N} T_{\rm IO}^{i} = \sigma_{\rm I} \otimes \mathds{1}_{\rm O}.
\end{align}
Furthermore, the dimension $d_{\rm E}$ of the memory system $\rm{E}$ in the Definition \ref{def:sc_tester} of a single-copy tester can always be chosen equal to $d_{\rm I}$, i.e., equal to the dimension of the input system.
\end{proposition}
Single-copy testers are thus characterised in terms of positive semidefinite matrices and affine constraints. Consequently, the maximal success probability $p_{\text{opt}}(\mathcal{E})$ for the discrimination of a channel ensemble $\mathcal{E} = \{q_i, \mathcal{C}_{\rm I \rightarrow O}^{i}\}_{i=1}^{N}$ can be efficiently computed by the semidefinite program
\begin{align} \label{eq:success_prob_optimisation}
   p_{\text{opt}}(\mathcal{E}) &= \max_{T_{\rm IO}^{i}} \sum_{i=1}^{N} q_{i} \, \Tr(T_{\rm IO}^{i}  C_{\rm IO}^{i}), 
\end{align}
where $\{T_{\rm IO}^{i}\}_{i=1}^{N}$ is a single-copy tester. 
For sake of convenience, we describe a proof of Proposition~\ref{prop:single_copy_phys_char} in Appendix \ref{app:phys_impl}. 

In the scenario in that quantum memory, quantified by the auxiliary dimension $d_{\rm E}$, is restricted, the characterisation of available tester elements is more involved.   

\subsection{Limited memory case}
In the case of limited memory, the auxiliary dimension $d_{\rm E}$ is assumed to have some value that is strictly smaller than $d_{\rm I}$. We see from Proposition~\ref{prop:single_copy_phys_char} that the case $d_{\rm I} = d_{\rm E}$ means no restriction on the memory system. This motivates the following definition of memory-restricted single-copy testers.
\begin{definition}[\blue{Memory-$d_{\rm E}$ single-copy tester}]
    A collection of positive operators $\{T_{\rm IO}\}_{i=1}^{N} \subset \mathcal{L}(\mathcal{H}_{\rm I} \otimes \mathcal{H}_{\rm O})$ is a memory-$d_{\rm E}$ single-copy tester if there is a quantum system $\rm{E}$ of dimension $d_{\rm E}$ such that 
    \begin{align}
        &T_{\rm IO}^{i} = \rho_{\rm IE} * (M_{\rm EO}^{i})^{T}  
    \end{align}
    where $\rho_{\rm IE}$ is a state and $\{M_{\rm EO}^{i}\}_{i=1}^{N}$ is a measurement. The corresponding circuit is:
\begin{center}
\resizebox{0.3\textwidth}{!}{\begin{tikzpicture}[x=1cm,y=1cm]

\filldraw [color=blue!60, fill=blue!5, very thick] 
(0,1.8) arc [start angle=90, end angle=270, x radius=.75, y radius=.9] 
node [pos=0.5,xshift=0.1cm,right,black] {\Large{$\rho$}};
\draw[very thick,color=blue!60](0,0) -- (0,1.8) ;

\draw[->,thick](0,1.4) --(3,1.4) node [pos=.5, above]{$\rm{E} \cong \mathbb{C}^{d_{\rm E}}$};
\draw[->,thick](0,.4)--(1,.4) node [pos=.5, below]{$\rm{I}$};

\filldraw [color=red!60, fill=red!5, very thick] 
(1,0) rectangle (2,.8) node [pos=0.5,black] {$C$};

\draw[->,thick](2,.4) --(3,.4) node [pos=.5, below]{$\rm{O}$};

\filldraw [color=blue!60, fill=blue!5, very thick] 
(3,0) arc [start angle=-90, end angle=90, x radius=.75, y radius=.9];
\draw[very thick,color=blue!60]
(3.0,1.8) -- node [right, black] {$M^{i}$}(3.0,0) ;

\end{tikzpicture}}
\end{center}
\end{definition}
By Proposition~\ref{prop:single_copy_phys_char},  every memory-$d_{\rm E}$ single-copy tester is also a single-copy tester as in Definition \ref{def:sc_tester}, but the opposite does not hold if $d_{\rm E} < d_{\rm I}$. Given a channel ensemble $\{q_i, \mathcal{C}_{\rm I \rightarrow O}^{i}\}_{i=1}^{N}$, we want to solve the optimisation problem
\begin{align} \label{eq:success_prob_optimisation_limited}
   p_{\text{opt},d_{\rm E}}(\mathcal{E}) &= \max_{T_{\rm IO}^{i}} \sum_{i=1}^{N} q_{i} \, \Tr(T_{\rm IO}^{i}  C_{\rm IO}^{i}), 
\end{align}
where the optimisation runs over all memory-$d_{\rm E}$ single-copy testers $T_{\rm IO}^{i}$.
Of particular interest are memory-less testers for which $d_{\rm E} = 1$. Here, the link product simply reduces to the usual tensor product so that tester elements $T_{\rm IO}^{i}$ are given by
\begin{equation}
\label{eq:memory_less_tester_element}
    T_{\rm IO}^{i} = \rho_{\rm I} \otimes (M^{i}_{\rm O})^T.
\end{equation} 

The amount of quantum memory quantified by the admissible environment dimension $d_{\rm E}$ is in general very important for the performance of a discrimination protocol. To illustrate, we consider the problem of discriminating the $d^2$ unitary evolutions provided by the clock-shift operators in dimension $d$, see Definition~\ref{def:important_channels}. We obtain the following result, see Appendix ~\ref{prop:heisenberg_weyl} for the proof.
\begin{Theorem}
\label{prop:clock_shift_memory_dependence}
    The maximal success probability $p_{\textnormal{opt}, d_{\rm E}}^{(d)}$ in discrimination of the uniform ensemble given by the $d$-dimensional clock-shift operators 
    using memory-$d_{\rm E}$ single-copy testers is given by 
    \begin{equation}
        p_{\textnormal{opt}, d_{\rm E}}^{(d)} = \min \left\{1,  \frac{d_{\rm E}}{d} \right\}.
    \end{equation}
\end{Theorem}
Clock-shift operators give rise to a remarkable dependence on quantum memory. For any fixed value of $d_{\rm E}$ and large $d$, the success-probability goes to zero, that is, $\lim_{d\to\infty} p_{\textnormal{opt}, d_{\rm E}}^{(d)} =0$. However, when $d_{\rm E}=d$, it holds that $p_{\textnormal{opt}, d_{E}}^{(d)} =1$ for any $d$. This shows that the performance ratio between unlimited and limited memory cases may be arbitrarily large. Hence, this problem instance demonstrates that the size of memory plays a major role in quantum channel discrimination.

Other examples where memory offers great advantage for channel discrimination may be found in \cite{Puzzuoli17}, where it is shown that for any dimension $d$, the two Werner-Holevo channels \cite{Werner2002}, i.e., channels whose Choi matrices are proportional to the projection onto the symmetric or antisymmetric subspace, is given by  $p_{d, d_{\rm E}} = \frac{1}{2} + \frac{\min \left\{d_{\rm E},d\right\}}{2d}$.

In the proof of Theorem~\ref{prop:clock_shift_memory_dependence}, we exploited the fact that the problem of discriminating uniformly distributed quantum channels with a group structure is a considerably simpler problem than general channel discrimination. 
Solving the problem of finding the optimal strategy in the setting of bounded memory for a general channel ensemble does not admit such a simple solution.

In the following, we provide a systematic technique for solving the problem \eqref{eq:success_prob_optimisation_limited} in general. It is first a simple but important observation that every channel discrimination problem with testers of limited memory is equivalent to an associated discrimination problem with memory-less testers.
\begin{observation}
\label{obs:memory-less-reduction}
    The discrimination problem of a channel ensemble $\{q_i, \mathcal{C}_{\rm I \rightarrow O}^{i}\}_{i=1}^{N}$ with memory-$d_{\rm E}$ single-copy testers is equivalent to the discrimination problem of the ensemble $\{q_i, \mathcal{C}_{\rm I \rightarrow O}^{i} \otimes \textnormal{id}_{\rm E}\}_{i=1}^{N}$ with memory-less single-copy testers. 
\end{observation}

Hence, every method to tackle the channel discrimination problem with memory-less testers can directly be applied to the case of limited memory.

The key idea of our approach comes from the fact that memory-less tester elements $T_{\rm IO}^{i}$, as a set of positive operators that can be decomposed as in Eq.~\eqref{eq:memory_less_tester_element}, can be identified with unnormalised quantum states $\underline{T}_{\rm NIO} := \bigoplus_{i=1}^{N}T_{\rm IO}^{i} = \sum_{i=1}^{N} \ketbra{i}_{\rm N} \otimes T_{\rm IO}^i$ where the states $\ket{i}$, often referred to as flag states, denote the standard basis of a system $\rm{N}$ with $d_{\rm N} = N$. We have
\begin{equation}
    \underline{T}_{\rm NIO} = P_{\rm I\leftrightarrow O} \, (\underline{M}_{\rm NO} \otimes \rho_{\rm I}) \, P_{\rm I \leftrightarrow O}^{\dagger}
\end{equation}
 where $P_{\rm I \leftrightarrow O}$ is the system permutation that exchanges input system $\rm{I}$ and output system $\rm{O}$, $\rho_{\rm I}$ is a state and $\underline{M}_{\rm NO} :=\bigoplus_{i=1}^{N} M_{\rm O}^i$ is the block-diagonal operator of measurement effects of $\{M_{\rm O}^{i}\}_{i=1}^{N}$. In this work, we use the convention of adding an underline to denote operators that have a direct sum structure, such as $\underline{M}_{\rm NO} $ and $\underline{T}_{\rm NIO}$.
 The operator $\underline{M}_{\rm NO}$ can be (after normalisation) seen as a quantum state that is subject to the affine constraint
 \begin{equation}
 \label{eq:single_copy_measurement_constraint}
     \Phi(\underline{M}_{\rm NO}) := \Tr_{\rm N}(\underline{M}_{\rm NO}) = \mathds{1}_{\rm O}.
 \end{equation}  
 The optimisation \eqref{eq:success_prob_optimisation_limited} over memory-less testers $p_{\text{opt},1}(\mathcal{\rm E})$ can then be explicitly written as 
 \begin{align}
     &\max_{\underline{T}} \, \Tr(\underline{C}_{\rm NIO} \; \underline{T}_{\rm NIO}) \label{eq:objective_function}\\ 
     &\text{s.t. } \underline{T}_{\rm NIO} = P_{\rm I \leftrightarrow O} \, (\underline{M}_{\rm NO} \otimes \rho_{\rm I}) \, P_{\rm I \leftrightarrow O}^{\dagger} \\
     & \quad\;\;\; \rho_{\rm I}, \underline{M}_{\rm NO} \geq 0 \\
     & \quad\;\; \Tr(\rho_{\rm I}) = 1, \Phi(\underline{M}_{\rm NO}) = \mathds{1}_{\rm O} \label{eq:extra_constraints_single_copy}
 \end{align}
 where $\underline{C}_{\rm NIO} :=\bigoplus_{i=1}^{N} q_{i} \, C_{\rm IO}^{i}$. Notice that the objective function \eqref{eq:objective_function} is linear, which implies that the optimisation may be equivalently performed over convex combinations of the operators $\underline{T}_{\rm NIO}$. 
As it will be discussed in Section \ref{sec:constraine_sep_main}, the problem corresponds to a so-called constrained separability problem that is maximisation of a linear functional over separable states whose local factors obey further affine constraints.

This insight allows us to state the following characterisation theorem of memory-less testers, which is our main result for the single-copy scenario. An explicit proof of Theorem~\ref{prop:symmetric_ext_tester} is presented in Appendix \ref{sec:mem_test_con_sep}.

\begin{Theorem} \label{prop:symmetric_ext_tester}
    A single-copy tester $\{T_{\rm IO}^{i}\}_{i=1}^{N} \subset \mathcal{L}(\mathcal{H}_{\rm I} \otimes \mathcal{H}_{\rm O})$ 
    is a convex combination of memory-less testers if and only if for all $k \in \mathbb{N}$ there exists (not necessarily memory-less) single-copy tester $\{T_{\textnormal{Sym}(\rm{I},k) \rm{O}}^{i}\}_{i=1}^{N} \subset \mathcal{L}(\mathcal{H}_{\rm I}^{\otimes k} \otimes \mathcal{H}_{\rm O})$ such that 
    \begin{equation}
        T_{\rm IO}^{i} = \Tr_{\rm{I}^{k-1}}(T_{\textnormal{Sym}(\rm{I},k) \rm{O}}^{i})
    \end{equation}
     where $\textnormal{Sym}(\rm{I},k)$ is the system that corresponds to the symmetric subspace of the Hilbert space $\mathcal{H}_{\rm I}^{\otimes k}$ of $k$ copies of the input system $\rm{I}$.
\end{Theorem}
 Theorem \ref{prop:symmetric_ext_tester} gives rise to a converging hierarchy of outer approximations to the set of memory-less single-copy testers, where each level is indexed by the number $k$ of copies of the input system $\rm{I}$. Later in Sec.~\ref{sec:constraine_sep_main}, we will see that Theorem~\ref{prop:symmetric_ext_tester} describes a special instance of a constrained symmetric extension hierarchy that provides a general method to find bounds to constrained separability problems. For any fixed $k$, the outer approximations only have affine and positive definite constraints, hence the problem of finding the optimal probability of success in the set of an outer approximation with parameter $k$ is a \blue{Semidefinite Program (SDP)}. SDPs describe a class of convex optimisation problems that is well studied, and in addition of respecting various properties, they can be numerically solved by efficient algorithms \cite{boyd2004convex}. Hence, in pragmatic terms, Theorem~\ref{prop:symmetric_ext_tester} provides a practical tool to compute upper bounds to the optimal success probability \eqref{eq:success_prob_optimisation_limited} of a memory-constrained channel discrimination problem. 
 
 \new{
 It is practical in the sense that sufficiently low levels $k$ of the hierarchy may be efficiently evaluated by an SDP implementation on a standard personal computer. 
 By counting the matrices and using the dimension
 of the symmetric subspace, the size of the SDP in the $k$-th level and memory dimension $d_{\rm E}$ amounts to $N$ positive semidefinite matrices that each have
 \begin{equation}
   \ell =  \underbrace{\binom{d_{\rm I}d_{\rm E} + k - 1}{k}}_{d_{\rm{Sym}(IE,k)}} * d_{\rm O} d_{\rm E}
 \end{equation}
 rows and columns where the left factor corresponds to the dimension of the symmetric subspace of $(\mathbb{C}^{d_{\rm I}d_{\rm E}})^{\otimes k}$ \cite{harrow2013}.
 By evaluating the binomial coefficient in terms of factorials, it follows that the number of scalar variables grows polynomially in $k$ with leading exponent $2(d_{\rm I} d_{\rm E}-1)$. 
 In particular, in the case of memory-less qubit testers, i.e., $d_{\rm I} = 2$ and $d_{\rm E} = 1$, the size of the SDP grows quadratically in $k$ allowing to compute the level $k=8$ using a personal computer within less than a minute.

 In a nutshell, the scaling of the computational complexity of our methods will behave as the scaling of the $k$-symmetric extension hierarchy for detecting entanglement (also known as DPS hierarchy)~\cite{PhysRevA.69.022308}. While our methods have strictly more constraints (we also impose the validity of the quantum tester), these additional constraints do not change the scaling of computational complexity. 
 }

 As an application of the hierarchy stated in Theorem~\ref{prop:symmetric_ext_tester}, we consider the discrimination of the channels given by the unitary square roots of three-dimensional clock-shift operators $\{\sqrt{X_{3}^{i} Z_{3}^{j}}\}_{i,j=0}^{2}$ defined in Definition~\ref{def:important_channels} with limited quantum memory, see Table \ref{tab:square_root}.
 With our methods, we are able to compute the optimal strategies with a gap of less than $10^{-2}$ between lower and upper bounds for all memory levels. For the computation of lower bounds, we applied the seesaw technique, see Appendix \ref{sec:seesaw} for details. The small size of the gap between the bounds is remarkable, as the considered set of channels has no underlying group structure that would simplify the problem. The source code used for our computations is available on a public repository \cite{gitcode}.

From a mathematical perspective, it also has to be noted that the presence of an extra constraint \eqref{eq:extra_constraints_single_copy} in the separability problem for memory-less testers is important. 
\begin{observation}
    \label{ob:separable_testers}
    The set of separable testers, i.e., testers $\{T_{\rm IO}^{i}\}_{i=1}^{N}$ of the form
    \begin{equation}
        T_{\rm IO}^{i} = \sum_{\lambda} q_{i \lambda} \; \rho_{\rm I}^{i \lambda} \otimes M_{\rm O}^{i \lambda}
    \end{equation}
    with any positive operators $\rho_{\rm I}^{i \lambda}$ and $M_{\rm O}^{i \lambda}$ and positive numbers $q_{i \lambda}$ does \textbf{not} coincide with the set of convex combinations memory-less testers, for which $\{M_{\rm O}^{i \lambda}\}_{i=1}^{N}$ is forced to be a measurement for each $\lambda$.
\end{observation}
 To see this, consider the discrimination of qubit to qubit channels with memory-less testers. The elements $T_{\rm IO}^{i}$ simply correspond to positive semidefinite operators on the system of two qubits. These are separable if and only if they have a positive partial transpose (PPT) \cite{Horodecki:1996nc}, see Appendix \ref{sec:ppt_criterion} for more details on the PPT criterion.

In a recent work in the context of Bayesian quantum metrology \cite{Bavaresco24}, the PPT condition has been used to described memory-less qubit testers. Although the PPT property is equivalent to separability in two-qubit systems, the following example demonstrates Remark \ref{ob:separable_testers} and shows the perhaps counter-intuitive fact that the PPT criterion is in general not enough for characterising memory-less qubit testers.  

 \begin{Example}
      \label{ex:sep_vs_consep}
      Consider the problem instance given by the uniform ensemble of an amplitude damping channel $\mathcal{C}_{\textnormal{ad}, \frac{2}{3}}$, a bit-flip channel $\mathcal{C}_{\textnormal{bf}, \frac{1}{3}}$, see Definition~\ref{def:important_channels}, and the identity $\textnormal{id}_{2}$. One observes that the maximal success probability with memory-less testers is given by $p = 0.556$. \new{This value has been obtained using the constrained symmetric extension hierarchy given by Theorem~\ref{prop:symmetric_ext_tester} up to level $k=8$, see Section \ref{sec:code}. The accuracy of this value has been verified with a seesaw optimisation.} 
      
      On the other hand, the relaxation of demanding that the tester elements are separable (equivalently PPT) leads to the outcome $p_{\textnormal{sep}}=0.562$ providing a clear separation of both optimisation problems. \new{This corresponds to the level $k=1$ in Theorem ~\ref{prop:symmetric_ext_tester}.}
 \end{Example}
 \new{
 The public repository \cite{gitcode} contains the script {\allowbreak ``amplitude\_damping\_bit\_flip\_identity\_success\\\_probability.jl''} that allows for a verification of the numerical results presented in Example \ref{ex:sep_vs_consep}}. 

\begin{table} 
    \centering
    \begin{tabular}{c|c c}
     & \multicolumn{2}{c}{optimal success probability} \\
    \cline{2-3}
     $d_{\rm E}$  & lower bound & upper bound\\
    \hline
     1 & 0.32627  & 0.32738 \;($k=4$)\\
     2 & 0.59416 & 0.60157 \;($k=1$)\\
     3 & 0.70126 & 0.70126
\end{tabular} 
\caption{Lower and upper bounds for the optimal success probabilities for discrimination of square roots of qutrit 
shift clock operators, see Definition~\ref{def:important_channels}, with $d_{\rm E}=1, 2, 3$. By Proposition~\ref{prop:single_copy_phys_char}, the value for $d_{\rm E} = 3$ can be exactly determined by a single SDP so that lower and upper bound coincide in this case. For $d_{\rm E} = 1,2$, lower bounds to the optimal success probability are obtained by a seesaw optimisation of the problem \eqref{eq:objective_function}, see Appendix \ref{sec:seesaw} for more details on seesaw optimisation. The upper bounds for $d_{\rm E}=1,2$ are obtained by applying the constrained symmetric extension hierarchy given by Theorem~\ref{prop:symmetric_ext_tester} to some finite order $k$ 
 combined with additional PPT constraints.}
 \label{tab:square_root}
\end{table}

\subsection{Constrained separability problems}
\label{sec:constraine_sep_main}
In the last section, we saw that the question of optimal channel discrimination with limited memory in the single-copy scenario can be described by the optimisation problem \eqref{eq:objective_function} over separable states with extra affine constraints that the factors in the separable decomposition have to obey. In the course of this work, we will see multiple appearances of variants of such problems since the multi-copy versions of the channel discrimination problem admit a similar mathematical description with properly defined generalised testers. 

As a general convex optimisation problem, \blue{constrained separability problems} have been introduced in Ref.~\cite{Berta2021} and they describe the optimisation of a linear functional $\Tr(F_{\rm AB}\, \cdot)$, specified by a hermitian operator $F \in \mathcal{L}(\mathcal{H}_{\rm A} \otimes \mathcal{H}_{\rm B})$,
that is of the form
\begin{align}
   &r_{\textnormal{opt}} =  \max_{\rho_{\rm AB}} \, \Tr(F_{\rm AB} \, \rho_{\rm AB}) \\
    &\textnormal{s.t.} \; \rho_{\rm AB} = \sum_{\lambda} p_{\lambda} \; \rho_{\rm A}^{\lambda} \otimes \rho_{\rm B}^{\lambda} \\
    & \forall \lambda: \; \rho_{\rm A}^{\lambda} \geq 0, \;  \Phi_{\rm A}(\rho_{\rm A}^{\lambda}) = a, \, \Tr(\rho_{\rm A}^{\lambda}) = 1 \label{eq:affine_a}\\
    & \forall \lambda: \; \rho_{\rm B}^{\lambda} \geq 0, \;  \Psi_{\rm B}(\rho_{\rm B}^{\lambda}) = b, \, \Tr(\rho_{\rm B}^{\lambda}) = 1 \label{eq:affine_b}\\
    & p_{\lambda} \geq 0, \; \sum_{\lambda} p_{\lambda} = 1. 
\end{align}

Here, $\Phi_{\rm A}:\mathcal{L}(\mathcal{H}_{\rm A})\rightarrow V_{\rm A}$ and $\Psi_{\rm B}:\mathcal{L}(\mathcal{H}_{\rm B})\rightarrow V_{\rm B}$ are fixed linear maps with vector spaces $V_{\rm{A}}$ and $V_{\rm{B}}$ and $a \in V_{\rm{A}}$ and $b \in V_{\rm{B}}$ are fixed vectors in the images of $\Phi_{\rm{A}}$ and $\Psi_{\rm{B}}$. The equations \eqref{eq:affine_a} and \eqref{eq:affine_b} describe affine constraints which make up the difference compared to the standard separability problem \cite{Guehne2009, Horodecki09}. The trace constraints are imposed to ensure that the outcome $r_{\rm opt}$ is finite. We saw an example of such an extra constraint already in Eq.~\eqref{eq:single_copy_measurement_constraint} which characterised the set of POVMs describing general quantum measurements. Notice that the problem, that the block matrices describing POVMs are not normalised as quantum states, can be easily taken care of by a suitable redefinition of the cost functional $F_{\rm AB}$.

Although separability problems are known to be generally hard to solve \cite{Gurvits2003}, lower and upper bounds can be obtained. A strong method for the computation of the latter comes from a converging hierarchy of semidefinite programs that is a generalisation of the converging hierarchy of symmetric extensions \cite{PhysRevA.69.022308} for the unconstrained case. Another important method for the computation of upper bounds is based on the positive partial transpose criterion (PPT) \cite{Horodecki:1996nc}, see Appendix \ref{sec:ppt_criterion} for details. The combination of both approaches leads to the \blue{PPT-constrained symmetric extension hierarchy}
\begin{align}
    &r_{k} = \max_{\rho_{\rm AB}} \, \Tr(F_{\rm AB} \, \rho_{\rm AB}) \label{eq:con_sym_opt_begin}\\
    &\textnormal{s.t.} \; \rho_{\rm AB} = \Tr_{{\rm{A}}^{k-1}}(\rho_{{\rm{A}}^{k}\rm{B}}) \\
    &\rho_{{\rm{A}}^{k}\rm{B}} \geq 0, \, \Tr(\rho_{\rm A^{k}B}) = 1 \\
    &\rho_{{\rm{A}}^{k}\rm{B}} = (U_{{\rm{A}}^{k}}^{\sigma} \otimes \mathds{1}_{\rm B}) \rho_{{\rm{A}}^{k}\rm{B}} ({U_{{\rm{A}}^{k}}^{\sigma}}^{\dagger} \otimes\mathds{1}_{\rm B}) \; \forall \sigma \\
    & \Tr_{\rm B}(\rho_{{\rm{A}}^{k}\rm{B}}) \otimes b = (\textnormal{id}_{{\rm{A}}^{k}} \otimes \Psi_{\rm B})\rho_{{\rm{A}}^{k}\rm{B}} \\
    &\Tr_{\rm AB}(\rho_{{\rm{A}}^{k}\rm{B}}) \otimes a = (\textnormal{id}_{{\rm{A}}^{k-1}} \otimes \Phi_{\rm A}) \Tr_{\rm B}(\rho_{{\rm{A}}^{k}\rm{B}}) \\
    &\forall l \in \{1,\dots,k\}: \;  {(\rho_{{\rm{A}}^{k}\rm{B}})}^{T_{{\rm{A}}^{l}}} \geq 0 \label{eq:con_sym_opt_end}
\end{align}
for $k\in \mathbb{N}$ where $\sigma \mapsto U_{{\rm{A}}^{k}}^{\sigma}$ is the standard unitary representation of the symmetric group $\mathcal{S}_{k}$ on $\mathcal{H}_{\rm A}^{\otimes k}$ \cite{harrow2013}. 
In the constrained setting, the hierarchy \eqref{eq:con_sym_opt_begin}-\eqref{eq:con_sym_opt_end} has been proven to converge~\cite{Berta2021} in the sense that $\lim_{k \rightarrow \infty} r_k = r_{\textnormal{opt}}$ 
based on the generalisations of the quantum De Finetti theorem~\cite{Berta2021,Caves2002, Christandl2007, aubrun2022monogamyentanglementcones,costa2024finettitheoremquantumcausal, Belzig2024}.
We review the problem of constrained separability in Appendix \ref{sec:con_sep} and give an elementary proof 
of the convergence of the constrained symmetric extensions that is based directly on the standard quantum De Finetti theorem \cite{Caves2002, PhysRevA.69.022308}.  

For the computation of lower bounds we use the seesaw technique in which initial local states are sampled, and local optimisations are performed in an alternating manner, see Appendix \ref{sec:seesaw} for more details. In the context of quantum metrology, such a technique has also been used in Ref.~\cite{kurdzialek2024}, where the authors make use of a tensor network representation to obtain non-trivial lower bounds for bounded memory scenarios in adaptive scenarios with up to $N=50$ copies of the given channel. Also, in Ref.~\cite{Bavaresco24}, the seesaw approach was used in the analysis of Bayesian parameter estimation in a memory-less scenario. 

The seesaw optimisation provides a valuable method for the computation of lower bounds. However, previously it has been difficult to determine the difference between the lower bound and the real optimum in the case that upper bounds cannot be computed or are relatively far away from the lower bound. To overcome this problem and to make quantitative statements about the quality of seesaw optimisation bounds, we apply the method of polytope approximations that has already been used for constructing local hidden variable models for entangled states \cite{Hirsch16,Cavalcanti2015LHV}, 
quantum steering \cite{Nguyen19, steinberg2023}, in the study of causal structures in quantum processes \cite{Nery2021} and for the standard quantum separability problem \cite{Ohst24}. The idea is to consider a fixed finite subset of operators $\{\widetilde{\rho}_{\rm A}^{\lambda}\}_{\lambda=1}^{m}$ of quantum states that satisfy $\Phi_{A}(\widetilde{\rho}_{\rm A}^{\lambda}) = a$ for all $\lambda$. The states $\{\widetilde{\rho}_{\rm A}^{\lambda}\}_{\lambda = 1}^{m}$ can be thought of as the set of initialisations chosen to start the seesaw algorithm. Geometrically, the convex hull $\mathcal{V}_{\rm A} = \textnormal{conv}(\{\widetilde{\rho}_{\rm A}^{\lambda}\}_{\lambda=1}^{m})$ is a polytope that approximates the constrained state space 
\begin{equation}
    Y_{\rm A} := \{\rho_{\rm A} \geq 0, \Tr(\rho_{\rm A}) = 1, \Phi_{\rm A}(\rho_{\rm A}) = a\}
\end{equation}
from inside, i.e., $\mathcal{V}_{\rm A} \subset Y_{\rm A}$. Given the polytope $\mathcal{V}_{\rm A}$, we consider \new{the optimisation problem}
\begin{align}
     r_{\mathcal{V}_{\rm A}} = \max_{\lambda} &\max_{\rho_{\rm B}}  \Tr(F_{\rm AB} \; \widetilde{\rho}_{\rm A}^{\lambda} \otimes \rho_{\rm B}) \\
    & \textnormal{s.t} \; \rho_{\rm B} \geq 0 \\ 
    &\Tr(\rho_{\rm B}) = 1, \Psi_{\rm B}(\rho_{\rm B}^{\lambda}) =  b
\end{align}
\new{that may be solved by a maximisation over the outcomes of finitely many semidefinite programs, out of which each contains a positive semidefinite variable matrix with $d_{\rm B}$ rows and columns.}
Notice that $r_{\mathcal{V}_{\rm A}} \leq r_{\rm opt}$ providing a lower bound to the optimisation problem. By considering ``outer polytopes'' that are given by a finite set of operators whose convex hull contains the constrained state space $Y_{\rm A}$, also upper bounds to $r_{\rm opt}$ can be computed in an analogous way~\cite{Hirsch16,Nguyen19,Nery2021,Ohst24, steinberg2023}. 

\begin{figure}
    \centering
    \includegraphics[width=0.9\linewidth]{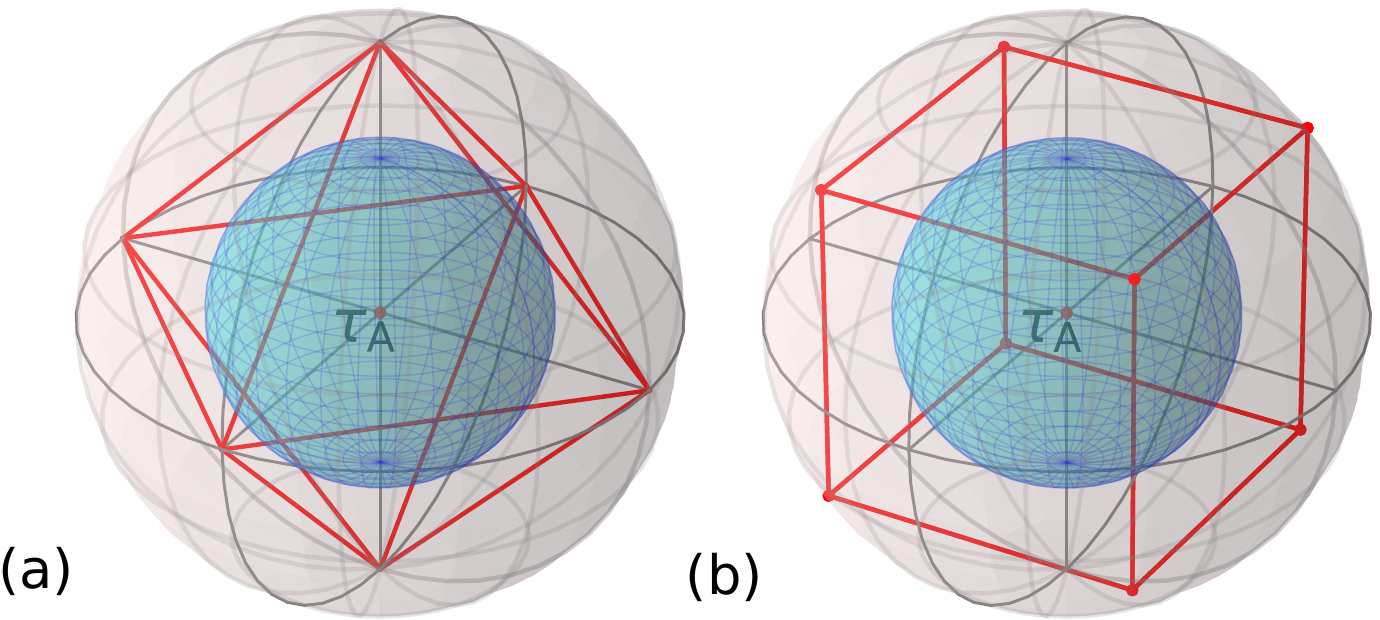}
    \caption{Simple inner polytopes $\mathcal{V}_{\rm A}$ of the unconstrained qubit state space (Bloch sphere) with reference state $\tau_{\rm A} = \frac{\mathds{1}_{2}}{2}$. (a) Polytope spanned by the $6$ eigenvectors of the three Pauli matrices. (b) Polytope spanned by the $8$ Bloch vectors $(\pm \frac{1}{\sqrt{3}}, \pm \frac{1}{\sqrt{3}}, \pm \frac{1}{\sqrt{3}})$. Although having a different number of vertices, both polytopes have the same approximation radius $l_{\tau_{\rm A}}(\mathcal{V}_{\rm A}) = \frac{1}{\sqrt{3}} \approx 0.577$ with respect to $\tau_{\rm A}$.}
    \label{fig:bloch_polytopes}
\end{figure}

In fact, an outer polytope can be obtained from the inner polytope $\mathcal{V}$, which allows for assessing the quality of the inner polytope itself. 
To this end, let $\tau_{A} \in \textnormal{int}(\mathcal{V}_{\rm A})$ be a fixed reference point in the interior of $\mathcal{V}_{\rm A}$ and let us define the approximation radius $l_{\tau_{\rm A}}(\mathcal{V}_{\rm A}) \in [0, 1]$ by 
\begin{align}
    l_{\tau_{\rm A}}(\mathcal{V}_{\rm A}) &:= \max_{t}  t \\
    &\textnormal{s.t} \; t \rho_{\rm A} + (1-t) \tau_{\rm A} \in \mathcal{V}_{\rm A} \; \forall \rho_{\rm A} \in Y_{\rm A}
\end{align}
Intuitively, the bigger $l_{\tau_{\rm A}}(\mathcal{V}_{\rm A})$ is, the better is the quality of the polytope approximation. The number $l_{\tau_{\rm A}}(\mathcal{V}_{\rm A})$ can be computed by a facet enumeration of the polytope $\mathcal{V}_{\rm A}$ and subsequent semidefinite programs that have to be computed for each facet\new{, see Appendix \ref{sec:seesaw}}. \new{Each of the SDPs contains a positive semidefinite variable matrix with $d_{\rm B}$ rows and columns.}  For instance in the case of the qubit state space, geometrically described by the Bloch sphere, and with the reference state given by the maximally mixed state, the approximation radius is just given by the largest sphere (centred at the origin) that is contained in the polytope, see Figure \ref{fig:bloch_polytopes}.

Further, let us define the value $f_{\tau_{\rm A}}$ as 
\begin{equation}
\label{eq:f_optimisation}
    f_{\tau_{\rm A}} := \min_{\rho_{\rm B} \in Y_{\rm B}} \Tr \left[\Tr_{\rm A}(F_{\rm AB} \, \tau_{\rm A} \otimes \mathds{1}_{\rm B}) \rho_{\rm B}\right]
\end{equation}
where the optimisation runs over all states $\rho_{\rm B}$ such that $\Psi_{\rm B}(\rho_{\rm B}) = b$ and $\Tr(\rho_{\rm B}) = 1$, that is, the constrained state space $Y_{\rm B}$. The value $f_{\tau_{\rm A}}$ only depends on the reference state $\tau_{A}$, especially it does not depend on the specific polytope $\mathcal{V}_{\rm A}$, and can be computed by a single SDP \new{involving a positive semidefinite matrix variable with $d_{\rm B}$ rows and columns.}

We can now state the following theorem that gives rise to a quantitative bound on the error in seesaw optimisation. The proof is presented in Appendix \ref{sec:seesaw} based on the construction of an outer polytope from an inner one.

\begin{Theorem}
    \label{thm:see_saw_bound}
    Let $(F_{\rm AB}, \Phi_{\rm A}, a, \Psi_{\rm B}, b)$ describe a constrained separability problem and let $\mathcal{V}_{\rm A}$ be an inner polytope of the constrained state space $Y_{\rm A}$ that has a non-zero approximation radius $l_{\tau_{\rm A}}(\mathcal{V}_{\rm A})$. Then for the polytope approximation $r_{\mathcal{V}_{\rm A}}$, we have the inequalities  
    \begin{equation}
        r_{\mathcal{V}_{\rm A}} \leq r_{\rm opt} \leq \frac{1}{l_{\tau_{\rm A}}(\mathcal{V}_{\rm A})} r_{\mathcal{V}_{\rm A}} + \frac{l_{\tau_{\rm A}}(\mathcal{V}_{\rm A})-1}{l_{\tau_{\rm A}}(\mathcal{V}_{\rm A})} f_{\tau_{\rm A}}
    \end{equation}
    In particular, the right-hand side of the inequality goes to $r_{\mathcal{V}_{\rm A}}$ if the approximation radius $l_{\tau_{\rm A}}(\mathcal{V}_{\rm A})$ goes to $1$.
\end{Theorem}

\section{Multi-copy channel discrimination without memory constraints}
\label{sec:multi_without_const}
In the multi-copy scenario, the channel discrimination protocol is strengthened by allowing to apply the unknown channel more than once within a single run. For simplicity, we focus on the situation in which two copies of the unknown channel are available, but all the following notions can be straightforwardly generalised to more copies. It is assumed that a channel is randomly drawn from the ensemble $\{ q_{i}, {(\mathcal{C}^i \otimes \mathcal{C}^i)}_{\rm \bf{I}\rightarrow \bf{O}} \}_{i=1}^{N}$ where $\bf{I}$ and $\bf{O}$ are the labels for the composite systems $\rm{I}_{1}\rm{I}_{2}$ and $\rm{O}_{1}\rm{O}_{2}$, respectively. With an appropriate generalisation of testers to the multi-copy scenario given by a set of positive operators $\{T_{\bf{IO}}^{i}\}_{i=1}^{N}$, we are interested in optimising the success probability 
\begin{align}
\label{eq:multi_suc_prob}
     p =  \sum_{i=1}^{N} q_{i} \Tr({(C^i \otimes C^i)}_{\bf{I}\bf{O}} T_{\bf{IO}}^{i})
\end{align}
where $C^{i}$ is again the Choi matrix of $\mathcal{C}^{i}$.  

In contrast to the single-copy scenario, there are multiple inequivalent strategies to perform the discrimination that lead to different definitions of the testers $\{T_{\bf{IO}}^{i}\}_{i=1}^{N}$. We will discuss these strategies first in the case of unlimited quantum memory. In the subsequent section, their respective memory restrictions will be discussed. 

\subsection{Parallel schemes without memory constraints}

The first possibility is to apply the two copies of the channel in parallel in the sense that they act as a product channel on a high-dimensional quantum state before a high-dimensional measurement is performed.

Mathematically, this means that one acts with the channel ${(\mathcal{C}^i \otimes \mathcal{C}^i)}_{\bf{I}\rightarrow \bf{O}} \otimes \text{id}_{\rm E}$ on a state $\rho_{\rm \mathbf{I}E}$ before some global measurement $M_{\rm E\bf{O}}^{i}$ is performed. Simply speaking, this strategy is the same as the single-copy scheme applied to the ensemble $\{\lambda_{i}, {(\mathcal{C}^i \otimes \mathcal{C}^i)}_{\rm \bf{I}\rightarrow \bf{O}} \}_{i=1}^{N}$ of product channels. Consequently, the associated tester elements are given by  $T_{\bf{IO}}^{i} = \rho_{\rm \textbf{I}E} * ({M_{\rm E\bf{O}}^{i}})^{T}$ which in analogy to single-copy testers \ref{def:sc_tester} justifies the following definition. 
\begin{definition}[Two-copy \blue{parallel tester} \cite{PhysRevA.80.022339}]
\label{def:unbounded_parallel}
     A collection of positive operators $\{T_{\bf{IO}}^{i}\}_{i=1}^{N} \subset \mathcal{L}(\mathcal{H}_{\rm I_{1}} \otimes\mathcal{H}_{\rm I_{2}} \otimes \mathcal{H}_{\rm O_{1}} \otimes \mathcal{H}_{\rm O_{2}})$ is called a two-copy parallel tester, if there is a quantum system $\rm{E}$, a state $\rho_{\rm \bf{I}E}$ and a measurement $M_{\rm E\bf{O}}^{i}$
     such that
     \begin{align}\label{eq:two_copy_parallel_tester_def}
     T_{\bf{IO}}^{i} = \rho_{\rm \mathbf{I}E} * {(M_{\rm E\bf{O}}^{i})}^{T}
\end{align}
The corresponding circuit looks as follows:
\begin{center}
    \resizebox{0.35\textwidth}{!}{\begin{tikzpicture}[x=1cm, y=1cm]
\clip (-1,-.2) rectangle (4,2.4);

\filldraw [color=blue!60, fill=blue!5, very thick] (0,2.2) arc [start angle=90, end angle=270, x radius=.75, y radius=1.1]
node [pos=.5,xshift=0.1cm, right, black] {\Large{$\rho$}};
\draw[very thick,color=blue!60](0,0) -- (0,2.2) ;

\draw[->,thick](0,.3) --(1,.3) node [pos=.5, below]{$\rm{I_2}$};
\draw[->,thick](0,1.1)--(1,1.1) node [pos=.5, above]{$\rm{I_1}$};
\draw[->,thick](0,1.9)--(3,1.9) node [pos=.5, above]{$\rm{E} \cong \mathbb{C}^{\infty}$};

\filldraw [color=red!60, fill=red!5, very thick]
(1,0) rectangle (2,.6) node [pos=0.5,black] {$\mathcal{C}$};
\filldraw [color=red!60, fill=red!5, very thick] 
(1,.8) rectangle (2,1.4) node [pos=0.5,black] {$\mathcal{C}$};

\draw[->,thick](2,.3) --(3,.3) node [pos=.5, below]{$\rm{O_2}$};
\draw[->,thick](2,1.1)--(3,1.1) node [pos=.5, above]{$\rm{O_1}$};

\filldraw [color=blue!60, fill=blue!5, very thick] (3,0) arc [start angle=-90, end angle=90, x radius=.75, y radius=1.1];
\draw[very thick,color=blue!60]
(3,0) -- node [right=0.01, black] {$M^{i}$}(3,2.2) ;

\end{tikzpicture}}
\end{center}
\end{definition}
Mathematically, it is fair to say that parallel testers are just single-copy testers as in Definition~\ref{def:sc_tester} and the only difference between them lies in the considered scenario in the sense that parallel testers allow the application on multiple copies of an unknown channel.
Consequently, similar to Proposition~\ref{prop:single_copy_phys_char}, parallel testers admit a characterisation in terms of positive semidefinite matrices and affine constraints via $\sum_{i=1}^{N} T_{\bf{IO}}^{i} = \sigma_{\bf{I}} \otimes \mathds{1}_{\bf{O}}$ where $\sigma_{\bf{I}}$ is a state.

\subsection{Adaptive schemes without memory constraints}
In adaptive schemes, one allows that one of the channels is used before the other one. This means that one admits quantum communication of the output of the first call to the input of the second call. In terms of tester elements, the adaptive scheme has been defined in Ref.~\cite{PhysRevA.80.022339}, see also \cite{Bavaresco21PRL}.  
\begin{definition}[Two-copy \blue{adaptive tester} \cite{PhysRevA.80.022339}]
\label{def:adaptive_tester}
     A collection of positive operators $\{T_{\bf{IO}}^{i}\}_{i=1}^{N} \subset \mathcal{L}(\mathcal{H}_{\rm I_{1}} \otimes\mathcal{H}_{\rm O_{1}} \otimes \mathcal{H}_{\rm I_{2}} \otimes \mathcal{H}_{\rm O_{2}})$  with $W_{\bf{IO}} = \sum_{i=1}^{N} T_{\bf{IO}}^{i}$ is called a two-copy adaptive tester, if there exist quantum systems $\rm{E}_1$ and $\rm{E}_2$ a state $\rho_{\rm I_{1}E_{1}}$, a channel $\mathcal{K}_{\rm E_{1}O_{1}\rightarrow I_{2}E_{2}}$ with Choi matrix $K_{\rm E_{1}O_{1}I_{2}E_{2}}$ and a measurement $\{M^{i}_{\rm E_{2}O_{2}}\}_{i=1}^{N}$ such that 
     \begin{align}\label{eq:two_copy_adaptive_tester_def}
    T_{\bf{IO}}^{i} = \rho_{\rm I_{1}E_{1}} * K_{\rm E_{1}O_{1}I_{2}E_{2}} * (M_{\rm E_{2}O_{2}}^{i})^{T}.
\end{align}
The corresponding circuit can be represented as follows:
\begin{center}
    \resizebox{0.45\textwidth}{!}{\begin{tikzpicture}[x=1cm,y=1cm]
\clip (-0.9,-.5) rectangle (7.7,2.5);

\filldraw [color=blue!60, fill=blue!5, very thick] 
(0,1.8) arc [start angle=90, end angle=270, x radius=.8, y radius=.9] 
node [pos=0.5,xshift=.1cm,right,black] {\Large{$\rho$}};
\draw[very thick,color=blue!60](0,0) -- (0,1.8) ;

\draw[->,thick](0,1.4) --(2.8,1.4) node [pos=.5, above]{$\rm{E_{1}} \cong \mathbb{C}^{\infty}$};
\draw[->,thick](0,.4)--(1,.4) node [pos=.5, below]{$\rm{I_1}$};

\filldraw [color=red!60, fill=red!5, very thick] 
(1,0) rectangle (1.8,.8) node [pos=0.5,black] {$\mathcal{C}$};

\draw[->,thick](1.8,.4) --(2.8,.4) node [pos=.5, below]{$\rm{O_1}$};

\filldraw [color=blue!60, fill=blue!5, very thick] 
(2.8,0) rectangle (3.8,1.8) node [pos=0.5,black] {$\mathcal{K}$};

\draw[->,thick](3.8,1.4) --(6.6,1.4) node [pos=.5, above]{$\rm{E_{2}} \cong \mathbb{C}^{\infty}$};
\draw[->,thick](3.8,.4)--(4.8,.4) node [pos=.5, below]{$\rm{I_2}$};

\filldraw [color=red!60, fill=red!5, very thick] 
(4.8,0) rectangle (5.6,.8) node [pos=0.5,black] {$\mathcal{C}$};

\draw[->,thick](5.6,.4) --(6.6,.4) node [pos=.5, below]{$\rm{O_2}$};

\filldraw [color=blue!60, fill=blue!5, very thick] 
(6.6,0) arc [start angle=-90, end angle=90, x radius=.8, y radius=.9];
\draw[very thick,color=blue!60]
(6.6,1.8) -- node [right, black] {$M^{i}$}(6.6,0) ;

\end{tikzpicture}}
\end{center}
\end{definition}
In physical terms, every two-copy adaptive scheme can be understood in the following way. The first copy of the unknown channel is applied on a part of potentially entangled state $\rho_{\rm I_{1}E_{1}}$ before a fixed global channel $\mathcal{K}_{\rm E_{1}O_{1}\rightarrow I_{2}E_{2}}$ transforms the resulting state.
After that, the second channel copy is applied on a subsystem of the output system of the channel $\mathcal{K}_{\rm E_{1}O_{1}\rightarrow I_{2}E_{2}}$ which is followed by a measurement $\{M^{i}_{\rm E_{2}O_{2}}\}_{i}$.

As for parallel testers, adaptive testers admit an SDP description that is practical from the point of view of optimisation tasks.
\begin{proposition}[\cite{PhysRevA.80.022339}]
\label{prop:adpative_charcterisation}
    A collection of positive operators $\{T_{\bf{IO}}^{i}\}_{i=1}^{N}$  with $W_{\bf{IO}} := \sum_{i=1}^{N} T_{\bf{IO}}^{i}$ is a two-copy adaptive tester if and only if 
    \begin{align}
    W &= R_{\rm \mathbf{I} O_1} \otimes \mathds{1}_{\rm O_{2}} \\
    \Tr_{\rm I_{2}}(R_{\rm \mathbf{I} O_1}) &= \sigma_{\rm I_{1}} \otimes \mathds{1}_{\rm O_1} 
\end{align}
where $R_{\rm \mathbf{I} O_1}$ is a positive operator and $\sigma_{\rm I_{1}}$ is a state.
\end{proposition} 

It follows directly from Definition~\ref{def:unbounded_parallel} and Proposition~\ref{prop:adpative_charcterisation} that every parallel tester is at the same time an adaptive tester. The inclusion is also proper in the sense that there exist adaptive testers which are not parallel. One can explicitly provide examples of channel discrimination tasks in which adaptive strategies outperform all parallel schemes \cite{Bavaresco22JMP}. An example is provided by the discrimination of the unitary operators given by the square roots of Pauli matrices $\{\mathds{1}, \sqrt{\sigma_{x}}, \sqrt{\sigma_{y}}, \sqrt{\sigma_{z}}\}$. On the one hand, their corresponding unitary channels can be perfectly discriminated with two copies in the adaptive scenario. In Ref.~\cite{Bavaresco22JMP}, the authors make use of computer assisted proof methods to show that in parallel scenarios, the maximal success probability is upper bounded by $p_{\text{par}}\leq \frac{9571}{10000}$. Here, we present a simple proof that $p_{\text{par}}< 1$ without making use of any computational support.
\begin{Theorem} 
\label{prop:shijun}
    It is not possible to perfectly discriminate the ensemble given by the unitary operators $\{\mathds{1}, \sqrt{\sigma_{x}}, \sqrt{\sigma_{y}}, \sqrt{\sigma_{z}}\}$ using two-copy parallel testers. 
\end{Theorem}
The proof of Theorem~\ref{prop:shijun} is presented in Appendix \ref{app:square_root_impossible}.
\section{Multi-copy channel discrimination with memory constraints and the role of classical memory}
\label{sec:multi_memory_const}
In analogy to the single-copy scenario, it is also important to consider the case of limited quantum memory for schemes involving multiple copies of the unknown channel. We will see that the role of quantum memory has to be analysed with special care in the investigation of adaptive strategies. The characterisation of tester elements depends on whether classical communication between consecutive measurements is permitted or not.
\subsection{Parallel schemes with memory constraints}
The similarity between parallel testers and single-copy testers carries over to the memory bounded scenario so that we define the parallel testers with limited quantum memory as follows.  
\begin{definition}[Memory-$d_{E}$ two-copy parallel tester]
    A collection of positive operators $\{T_{\bf{IO}}^{i}\}_{i=1}^{N} \subset \mathcal{L}(\mathcal{H}_{\rm I_{1}} \otimes\mathcal{H}_{\rm I_{2}} \otimes \mathcal{H}_{\rm O_{1}} \otimes \mathcal{H}_{\rm O_{2}})$ is a memory-$d_{\rm E}$ two-copy parallel tester if there is a quantum system $\rm{E}$ of dimension $d_{\rm E}$ such that 
    \begin{align}
        &T_{\bf{IO}}^{i} = \rho_{\rm \mathbf{I}E} * (M_{\rm E\mathbf{O}}^{i})^{T}  
    \end{align}
    where $\rho_{\rm \mathbf{I} \rm{E}}$ is s state and $\{M_{\rm \rm{E}\mathbf{O}}^{i}\}_{i=1}^{N}$ is a measurement. 
    \begin{center}
        \resizebox{0.35\textwidth}{!}{\begin{tikzpicture}[x=1cm, y=1cm]
\clip (-1,-.2) rectangle (4,2.4);

\filldraw [color=blue!60, fill=blue!5, very thick] (0,2.2) arc [start angle=90, end angle=270, x radius=.75, y radius=1.1]
node [pos=.5,xshift=0.1cm, right, black] {\Large{$\rho$}};
\draw[very thick,color=blue!60](0,0) -- (0,2.2) ;

\draw[->,thick](0,.3) --(1,.3) node [pos=.5, below]{$\rm{I_2}$};
\draw[->,thick](0,1.1)--(1,1.1) node [pos=.5, above]{$\rm{I_1}$};
\draw[->,thick](0,1.9)--(3,1.9) node [pos=.5, above]{$\rm{E} \cong \mathbb{C}^{d_{\rm{E}}}$};

\filldraw [color=red!60, fill=red!5, very thick]
(1,0) rectangle (2,.6) node [pos=0.5,black] {$\mathcal{C}$};
\filldraw [color=red!60, fill=red!5, very thick] 
(1,.8) rectangle (2,1.4) node [pos=0.5,black] {$\mathcal{C}$};

\draw[->,thick](2,.3) --(3,.3) node [pos=.5, below]{$\rm{O_2}$};
\draw[->,thick](2,1.1)--(3,1.1) node [pos=.5, above]{$\rm{O_1}$};

\filldraw [color=blue!60, fill=blue!5, very thick] (3,0) arc [start angle=-90, end angle=90, x radius=.75, y radius=1.1];
\draw[very thick,color=blue!60]
(3,0) -- node [right=0.01, black] {$M^{i}$}(3,2.2) ;

\end{tikzpicture}}
    \end{center}
\end{definition}
Notably, this definition is completely analogous to the one of single-copy testers and there are no new effects that would have to be discussed.

In the parallel scheme, it is known \cite{Taranto2024characterising} that classical communication can be modelled by convex combination of testers which has no practical effect due to the linearity of the optimisation objective \eqref{eq:multi_suc_prob}.
On the other hand, imposing memory restrictions in adaptive strategies is more subtle, as we will explore in the following. The reason for this is that the possibility of classical communication within the respective channel application has to be considered.  

\subsection{Adaptive schemes with memory constraints}
We first consider the case in which the environment systems $\rm{E}_1$ and $\rm{E}_2$ are bounded and where no additional classical memory is accessible within the test. Such an assumption has been made in Ref.~\cite{Vieira2024} for the study of temporal correlations in open system dynamics. We define the corresponding testers as follows.
\begin{definition}[Memory-($d_{\rm E_1},d_{\rm E_{2}})$ two-copy \blue{adaptive tester without additional classical memory}]
\label{def:memory_adaptive_tester_no_cc}
    A collection of positive operators $\{T_{\bf{IO}}^{i}\}_{i=1}^{N} \subset \mathcal{L}(\mathcal{H}_{\rm I_{1}} \otimes\mathcal{H}_{\rm O_{1}} \otimes \mathcal{H}_{\rm I_{2}} \otimes \mathcal{H}_{\rm O_{2}})$ is a memory-$(d_{\rm E_1},d_{\rm E_{2}})$ two-copy adaptive tester without additional classical memory if there are quantum systems $\rm{E}_{1}$ and $\rm{E}_{2}$ of dimensions $d_{\rm E_{1}}$ and $d_{\rm E_{2}}$ such that 
    \begin{align}
    \label{eq:ad_no_cc}
        &T_{\bf{IO}}^{i} = \rho_{\rm I_{1}E_{1}} * K_{\rm E_{1}O_{1}I_{2}E_{2}} * (M_{\rm E_{2} O_{2}}^{i})^{T} 
    \end{align}
    where $\rho_{\rm I_{1}E_{1}}$ is a state, $K_{\rm E_{1}O_{1}I_{2}E_{2}}$ is the Choi-matrix of a quantum channel $\mathcal{K}_{\rm E_{1}O_{1} \rightarrow I_{2}E_{2}}$ and $\{M_{\rm E_{2} O_{2}}^{i}\}_{i=1}^{N}$ is a measurement. The corresponding circuit is:
    \begin{center}
        \resizebox{0.45\textwidth}{!}{\begin{tikzpicture}[x=1cm,y=1cm]
\clip (-0.9,-.5) rectangle (7.7,2.5);

\filldraw [color=blue!60, fill=blue!5, very thick] 
(0,1.8) arc [start angle=90, end angle=270, x radius=.8, y radius=.9] 
node [pos=0.5,xshift=.1cm,right,black] {\Large{$\rho$}};
\draw[very thick,color=blue!60](0,0) -- (0,1.8) ;

\draw[->,thick](0,1.4) --(2.8,1.4) node [pos=.5, above]{$\rm{E_1} \cong \mathbb{C}^{d_{\rm{E_1}}}$};
\draw[->,thick](0,.4)--(1,.4) node [pos=.5, below]{$I_1$};

\filldraw [color=red!60, fill=red!5, very thick] 
(1,0) rectangle (1.8,.8) node [pos=0.5,black] {$\mathcal{C}$};

\draw[->,thick](1.8,.4) --(2.8,.4) node [pos=.5, below]{$\rm{O_1}$};

\filldraw [color=blue!60, fill=blue!5, very thick] 
(2.8,0) rectangle (3.8,1.8) node [pos=0.5,black] {$\mathcal{K}$};

\draw[->,thick](3.8,1.4) --(6.6,1.4) node [pos=.5, above]{$\rm{E_2} \cong \mathbb{C}^{d_{\rm{E_2}}}$};
\draw[->,thick](3.8,.4)--(4.8,.4) node [pos=.5, below]{$\rm{I_2}$};

\filldraw [color=red!60, fill=red!5, very thick] 
(4.8,0) rectangle (5.6,.8) node [pos=0.5,black] {$\mathcal{C}$};

\draw[->,thick](5.6,.4) --(6.6,.4) node [pos=.5, below]{$\rm{O_2}$};

\filldraw [color=blue!60, fill=blue!5, very thick] 
(6.6,0) arc [start angle=-90, end angle=90, x radius=.8, y radius=.9];
\draw[very thick,color=blue!60]
(6.6,1.8) -- node [right, black] {$M^{i}$}(6.6,0) ;

\end{tikzpicture}}
    \end{center}
\end{definition}
Quantum memory constraints without possible replacement by classical communication are a strong restriction. Notice that the tester definition in Eq.~\eqref{eq:ad_no_cc} implies that the action of the measurement $\{M_{\rm EO_{2}}^{i}\}_{i}$ on the system $\rm{E}\rm{O}_2$ is completely rigid and can not change due to any event within the test. 
We consequently obtain the following general bound.
\begin{Theorem}
\label{prop:adpt_gen_bound}
    Let $\{\frac{1}{N}, \mathcal{C}^{i}_{\rm I\rightarrow O}\}_{i=1}^{N}$ be a uniform channel ensemble. The optimal success probability $p$ for discriminating this ensemble using memory-($d_{\rm E_1},d_{\rm E_{2}})$ two-copy adaptive testers without additional classical memory is bounded as 
    \begin{equation}
        p \leq \min \left\{\frac{d_{\rm O} d_{\rm E_{2}}}{N}, 1\right\}.
    \end{equation}
\end{Theorem} 
\begin{proof}
    The bound follows from the fact that for any fixed choice of $\rho_{\rm I_{1}E}$ and $\mathcal{K}_{\rm E_{1}O_{1} \rightarrow I_{2}E_{2}}$, finding the optimal measurement $\{M_{\rm E_{2}O_{2}}^{i}\}_{i=1}^{N}$ is a state discrimination task of a uniform ensemble of states $\{\widetilde{\rho}_{\rm E_2 O}^{i}\}_{i=1}^{N}$. For every measurement choice $M_{\rm E_2 O}^{i}$,  one then has 
    \begin{align}
        &\sum_{i=1}^{N} \Tr(\widetilde{\rho}^i M^i) \leq  \sum_{i=1}^{N} \sqrt{\Tr((\widetilde{\rho}^{i})^2) \Tr({(M^i)}^2)} \\
        &\leq \sum_{i=1}^{N} \Tr(M^i) = d_{\rm O} d_{\rm E_{2}}
    \end{align}
    where the Cauchy-Schwarz inequality and the inequality $\Tr(A^2)\leq \Tr(A)^2$ for positive operators $A$ have been used. This gives rise to the bound $p \leq \frac{d_{\rm O} d_{\rm E_{2}}}{N}$ for the uniform state ensemble.
\end{proof}

As a direct consequence of Theorem~\ref{prop:adpt_gen_bound}, one can observe the following for the memory-less discrimination of the clock-shift operators.
\begin{observation}
\label{prop:clock-shift:no_go}
    The maximal success probability $p$ for channel discrimination of the uniform ensemble of the $d^2$ qudit clock shift unitaries $\{X_d^i Z_d^j\}_{i,j=0}^{d-1}$ with memory-less adaptive testers without classical memory is upper-bounded by $\frac{1}{d}$. 
\end{observation}

In particular, the optimal success probability decays to zero for $d \rightarrow \infty$. On the other hand, in the next section in Theorem~\ref{prop:perfect_clock_shift_discr}, we will see that a perfect discrimination is possible if classical memory is permitted and quantum memory remains inaccessible at the same time. Hence, classical memory can greatly enhance the performance in some channel discrimination tasks.

\subsection{The role of classical memory}
From a practical perspective, the memory restriction in Definition~\ref{def:memory_adaptive_tester_no_cc} in the adaptive setting of the preceding section is arguably too drastic. It might well be technically possible to feed-forward the first measurement result to the second measurement setting during the test, independent of whether entanglement is present in the state at some point or not. For this reason, we introduce the notion of adaptive testers with classical memory.

\subsubsection{Adaptive schemes with classical communication}
Technically, the possibility of classical communication in the setting of bounded quantum memory can be included by replacing the transformation channel $\mathcal{K}_{\rm E_{1}O_{1} \rightarrow I_{2}E_{2}}$ between the channel applications by a quantum instrument $\{\mathcal{K}_{\rm E_{1}O_{1} \rightarrow I_{2}E_{2}}^{j}\}_{j=1}^{L}$ whose classical outcome $j$ controls the final measurement $\{M^{i|j}_{\rm E_{2} O_{2}}\}_{i=1}^{N}$ where $j$ labels the measurement setting and $i$ labels the outcome. Similar constructions have been considered in Ref.~\cite{Nakahira21} in the context of process discrimination in general restricted scenarios. The precise definition that we assign to classical memory in adaptive schemes is the following. 
\begin{definition}[Memory-$(d_{\rm E_{1}},d_{\rm E_{2}})$ two-copy \blue{adaptive tester with classical memory}]
\label{def:adpt_tester_with_classical}
    A collection of positive operators $\{T_{\bf{IO}}^{i}\}_{i=1}^{N} \subset \mathcal{L}(\mathcal{H}_{\rm I_{1}} \otimes\mathcal{H}_{\rm O_{1}} \otimes \mathcal{H}_{\rm I_{2}} \otimes \mathcal{H}_{\rm O_{2}})$ is a memory-$(d_{\rm E_1},d_{\rm E_{2}})$ two-copy adaptive tester admitting classical memory if there are quantum systems $\rm{E}_{1}$ and $\rm{E}_{2}$ of dimensions $d_{\rm E_{1}}$ and $d_{\rm E_{2}}$ such that 
    \begin{align}
        T_{\bf{IO}}^{i} = \sum_{j=1}^{L} \rho_{\rm I_{1}E_{1}} * K_{\rm E_{1}O_{1}I_{2}E_{2}}^{j}* {(M_{\rm E_{2} O_{2}}^{i|j})}^{T} 
    \end{align}
    where $\rho_{\rm I_{1}E_{1}}$ is a state, $\{K^{j}_{\rm E_{1}O_{1}I_{2}E_{2}}\}_{j=1}^{L}$ are the Choi-matrices of elements of a quantum instrument $\{\mathcal{K}^{j}_{\rm E_{1}O_{1} \rightarrow I_{2}E_{2}}\}_{j=1}^{L}$ of size $L \in \mathbb{N}$ and $\{M_{\rm E_{2} O_{2}}^{i|j}\}_{i=1}^{N}$ are measurements for all $j$.
    \begin{center}
        \resizebox{0.45\textwidth}{!}{\begin{tikzpicture}[x=1cm,y=1cm]
\clip (-1.0,-.5) rectangle (7.7,2.5);

\filldraw [color=blue!60, fill=blue!5, very thick] 
(0,1.8) arc [start angle=90, end angle=270, x radius=.9, y radius=.9] 
node [pos=0.5,xshift=.1cm,right,black] {\Large{$\rho$}};
\draw[very thick,color=blue!60](0,0) -- (0,1.8) ;

\draw[->,thick](0,1.4) --(2.8,1.4) node [pos=.5, above]{$\rm{E_1} \cong \mathbb{C}^{d_{\rm{E_1}}}$};
\draw[->,thick](0,.4)--(1,.4) node [pos=.5, below]{$\rm{I_1}$};

\filldraw [color=red!60, fill=red!5, very thick] 
(1,0) rectangle (1.8,.8) node [pos=0.5,black] {$\mathcal{C}$};

\draw[->,thick](1.8,.4) --(2.8,.4) node [pos=.5, below]{$\rm{O_1}$};

\filldraw [color=blue!60, fill=blue!5, very thick] 
(2.8,0) rectangle (3.8,1.8) node [pos=0.5,black] {$\mathcal{K}^{j}$};

\draw [->, dashed, thick] (3.8,1.8) to [bend left=30] (6.6,1.8);

\draw[->,thick](3.8,1.4) --(6.6,1.4) node [pos=.5, above]{$\rm{E_2} \cong \mathbb{C}^{d_{\rm{E_2}}}$};
\draw[->,thick](3.8,.4)--(4.8,.4) node [pos=.5, below]{$\rm{I_2}$};

\filldraw [color=red!60, fill=red!5, very thick] 
(4.8,0) rectangle (5.6,.8) node [pos=0.5,black] {$\mathcal{C}$};

\draw[->,thick](5.6,.4) --(6.6,.4) node [pos=.5, below]{$\rm{O_2}$};

\filldraw [color=blue!60, fill=blue!5, very thick] 
(6.6,0) arc [start angle=-90, end angle=90, x radius=.9, y radius=.9];
\draw[very thick,color=blue!60]
(6.6,1.8) -- node [right, black, xshift=-0.1cm] {\small{$M^{i|j}$}} (6.6,0);

\end{tikzpicture}}
    \end{center}
    The amount of classical memory that is needed to implement such an adaptive tester can be quantified by $\log_2(L)$ where $L$ is the number of non-vanishing instrument elements. 
\end{definition} 

In the last subsection, we saw in Remark \ref{prop:clock-shift:no_go} how the discrimination success decayed for clock-shift operators and rising dimensions in the completely memory-less scenario. On the other extreme, by allowing for classical memory, a perfect discrimination is possible in all dimensions.

\begin{Theorem}
\label{prop:perfect_clock_shift_discr}
    The clock-shift-unitaries $\{X_d^i Z_d^j\}_{i,j=0}^{d-1}$, see Definition~\ref{def:important_channels}, can be perfectly discriminated with a two-copy adaptive tester without quantum memory but with classical memory.
\end{Theorem}
\begin{proof}
    We explicitly explain the strategy. First, the state $\rho= \ketbra{0}{0}$ is prepared. Since $Z_d^j$ leaves $\rho$ invariant for all $j$, the unitary channel with the unitary $X_d^i Z_d^j$ transforms $\rho$ into the state $\ketbra{i}{i}$ and $i$ is deterministically measured in a computational basis measurement. Let now $\ket{\psi_{i}} := \sum_{j=0}^{d-1} \omega^{i\cdot j} \ket{j}$ be the Fourier basis. Independent of the first measurement outcome, in the second run one prepares $\Tilde{\rho} = \ketbra{\psi_0}{\psi_0}$ so that $X_d^i Z_d^j \ket{\psi_0} = \frac{\omega^{-i}}{\sqrt{d}} \ket{\psi_j}$ which deterministically leads to the outcome $j$ in a Fourier basis measurement. By classically memorising the first measurement result $i$, all the $d^2$ clock-shift operators can hence be distinguished. 
\end{proof}
The statement of Theorem~\ref{prop:perfect_clock_shift_discr}
demonstrates how subtle the advantage of memory usage in quantum information processing, especially for the problem of channel discrimination, is. \new{We see that the amount of classical memory, quantified by $L$, may play an important role for the optimal discrimination success. In view of Theorem \ref{prop:adpt_gen_bound}, there are instances in which a sufficient amount of classical memory can compensate a lack of quantum memory quantified by the environment dimensions $d_{\rm E_{1}}$ and $d_{\rm E_{2}}$}. This demands for finding methods to efficiently characterise those memory restrictions.

Similar to single-copy testers, memory restricted adaptive testers admit a description in terms of constrained separable states. 
\subsection{Optimisation over memory-constrained adaptive testers as constrained separability problem}
Here, we show how the computation of maximal success probabilities in a channel discrimination task, using adaptive and memory restricted testers, can be cast as a constrained separability problem. First, with an argument similar to Remark \ref{obs:memory-less-reduction}, one can restrict to the quantum memory-less scenario, i.e., $d_{\rm E_1} = d_{\rm E_2} = 1$ with a suitable redefinition of the channel ensemble. The tester elements $T_{\bf{IO}}^{i}$ from Definition \ref{def:adpt_tester_with_classical} then take the form
\begin{equation}
    T_{\bf{IO}}^{i} = \sum_{j=1}^{L} \rho_{\rm I_1} \otimes K_{\rm O_1 I_2}^{j} \otimes M_{\rm O_2}^{i|j}
\end{equation}
with $i \in \{1,\dots,N\}$. With a permutation $P$ of the systems $\rm{N,I_1,I_2,O_1}$ and $\rm{O}_2$, we can write the corresponding tester $\underline{T}_{\rm N\bf{IO}} := \bigoplus_{i=1}^{N} T_{\rm \bf{IO}}^{i}$ as
\begin{equation}
    \underline{T}_{\rm N\bf{IO}} = P \left(\Pi\left[ \, \underline{\underline{M}}_{\rm LNO_{2}} \otimes  \underline{K}_{\rm L'O_{1}I_{2}} \otimes \rho_{\rm I_{1}} \right] \,  \right)P^{\dagger}
\end{equation}
where $\underline{\underline{M}}_{\rm LNO_{2}} := \bigoplus_{j=1}^{L} \bigoplus_{i=1}^{N} M_{\rm O_2}^{i|j}$, $\underline{K}_{\rm L'O_1 I_2} := \bigoplus_{j=1}^{L} K_{\rm O_1 I_2}^{j}$ and $\Pi$ is the projection
\begin{equation}
\label{eq:con_sep_adpt_projection}
    \Pi[\cdot] := \sum_{j=1}^{L}\Tr_{\rm LL'}[(\ketbra{jj}{jj}_{\rm LL'} \otimes \mathds{1}_{\rm N \mathbf{IO}})\; (\cdot)]
\end{equation}
that is needed for discarding all terms of the form $M_{\rm{O}_2}^{i|j} \otimes K_{\rm O_1 I_2}^{k} \otimes \rho_{\rm I_{1}}$ where $k \neq j$. 

Hence, the operator $\underline{T}_{\rm N\mathbf{IO}}$ can be seen as a tripartite fully separable state whose factors obey additional affine-linear constraints. Explicitly, these affine constraints are 
\begin{align}
     &\Tr_{\rm L'I_{2}}(\underline{K}_{\rm L' O_{1}I_{2}}) = \mathds{1}_{\rm O_1} \\
     &\Tr_{\rm LN}[(\ketbra{j}{j}_{\rm L} \otimes \mathds{1}_{\rm NO_{2}}) \underline{\underline{M}}_{\rm LNO_{2}}] = \mathds{1}_{\rm O_2} \; \forall j.
\end{align}
The case of no classical memory then simply corresponds to setting $L=1$. 

\section{Classically adaptive schemes}
\label{sec:class_adp}
Remarkably, the instrument $\{\mathcal{K}^j\}_{j}$ used in the proof of Theorem~\ref{prop:perfect_clock_shift_discr} is rather special in the sense that it is entanglement-breaking. This motivates a further restriction of adaptive tester to classically adaptive testers.
\subsection{Classically adaptive testers}
In contrast to general adaptive schemes, classically adaptive schemes describe the situation in which there is no quantum communication by a transfer of quantum states between the uses of the channel. The strategy can then be described by forcing the general instrument $\{\mathcal{K}^j_{\rm A \rightarrow B}\}_{j=1}^{L}$ to be entanglement-breaking or equivalently to be a measure-and-prepare instrument \cite{Horodecki2003}. Such instruments have elements of the form $\mathcal{K}_{\rm A \rightarrow B}^j(\cdot) = \Tr(F_{\rm A}^{j} \;\; \cdot) \sigma_{\rm B}^{j}$ where $\{F_{\rm A}^{j}\}_{j=1}^{L}$ is a measurement and $\sigma_{\rm B}^{j}$ are states. In terms of the Choi isomorphism, this implies that all Choi matrices $K_{\rm AB}^{j}$ are separable. 

Classically adaptive testers for which the instrument $\{\mathcal{K}_{\rm E_1 O_1 \rightarrow I_2 E_2}^j\}_{j=1}^{L}$ is forced to be entanglement breaking are, using single-copy testers, defined as follows. First, a single-copy tester is applied to the first copy of the channel. After that, a second tester which may depend on the first outcome is chosen to be applied on the second copy.
Such types of quantum testers also appeared in Ref.~\cite{Nakahira21} in the context of discrimination of quantum process in restricted scenarios.
\begin{definition}[Two-copy \blue{classically adaptive tester}]
\label{def:class_adp}
     A collection of positive operators $\{T_{\bf{IO}}^{i}\}_{i=1}^{N} \subset \mathcal{L}(\mathcal{H}_{\rm I_{1}} \otimes\mathcal{H}_{\rm O_{1}} \otimes \mathcal{H}_{\rm I_{2}} \otimes \mathcal{H}_{\rm O_{2}})$ is called a two-copy classically adaptive tester, if
\begin{align}\label{eq:two_copy_classically_adaptive_tester_def}
   T_{\bf{IO}}^{i} &= \sum_{j=1}^{L} R_{\rm I_{1}O_{1}}^{j} \otimes S_{\rm I_{2}O_{2}}^{i|j}
\end{align}
where $\{R_{\rm I_{1}O_{1}}^{i}\}_{i=1}^{N}$ and $\{S_{\rm I_{2}O_{2}}^{i|j}\}_{i=1}^{N}$ are single-copy testers for all $j$.  
\begin{center}
    \resizebox{0.45\textwidth}{!}
    {\begin{tikzpicture}[x=1cm,y=1cm]
\clip (-1.0,-.5) rectangle (8.8,2.5);

\filldraw [color=blue!60, fill=blue!5, very thick] 
(0,1.8) arc [start angle=90, end angle=270, x radius=.9, y radius=.9] 
node [pos=0.5,xshift=0.1cm,right,black] {\Large{$\rho$}};
\draw[very thick,color=blue!60](0,0) -- (0,1.8) ;

\draw[->,thick](0,1.4) --(2.8,1.4) node [pos=.5, above]{$\rm{E_1} \cong \mathbb{C}^{d_{\rm{E_1}}}$};
\draw[->,thick](0,.4)--(1,.4) node [pos=.5, below]{$\rm{I_1}$};

\filldraw [color=red!60, fill=red!5, very thick] 
(1,0) rectangle (1.8,.8) node [pos=0.5,black] {$\mathcal{C}$};

\draw[->,thick](1.8,.4) --(2.8,.4) node [pos=.5, below]{$\rm{O_1}$};

\filldraw [color=blue!60, fill=blue!5, very thick] 
(2.8,0) arc [start angle=-90, end angle=90, x radius=.75, y radius=.9];
\draw[very thick,color=blue!60]
(2.8,1.8) -- node [right,xshift=-0.05cm, black] {$F^j$}(2.8,0);

\draw[->,dashed,thick](3.6,.9) --(4.2,.9) node [pos=.5, below]{$j$};
\draw [->, dashed, thick] (5.1, 2.0) to [bend left=30] (7.8,1.8);

\filldraw [color=blue!60, fill=blue!5, very thick] 
(5,1.8) arc [start angle=90, end angle=270, x radius=.8, y radius=.9] 
node [pos=0.5,xshift=0.05cm,right,black] {\Large{$\sigma^j$}};
\draw[very thick,color=blue!60](5,0) -- (5,1.8) ;

\draw[->,thick](5,1.4) --(7.8,1.4) node [pos=.5, above]{$\rm{E_2} \cong \mathbb{C}^{d_{\rm{E_2}}}$};
\draw[->,thick](5,.4)--(6,.4) node [pos=.5, below]{$\rm{I_2}$};

\filldraw [color=red!60, fill=red!5, very thick] 
(6,0) rectangle (6.8,.8) node [pos=0.5,black] {$\mathcal{C}$};

\draw[->,thick](6.8,.4) --(7.8,.4) node [pos=.5, below]{$\rm{O_2}$};

\filldraw [color=blue!60, fill=blue!5, very thick] 
(7.8,0) arc [start angle=-90, end angle=90, x radius=.9, y radius=.9];
\draw[very thick,color=blue!60]
(7.8,1.8) -- node [right,xshift=-0.1cm, black] {\small{$M^{i|j}$}}(7.8,0) ;

\draw [dashed, color=blue!80, very thick]
(2.7,-0.2) -- (2.7,2) -- (5.1,2)--(5.1,-0.2) -- cycle node [pos=0.5, above]{\Large{$\mathcal{K}^{j}$}};

\end{tikzpicture}}
\end{center}
\end{definition}
It can be checked that every classically-adaptive tester is also an adaptive tester, but the opposite statement is not true. We derived the following relation between the sets of testers, see Fig.~\ref{fig:adaptive sets} for an illustration and Appendix \ref{sec:prop_cadpt} for the proof. 
\begin{Theorem}
\label{prop:parallel_vs_cadpt}
    The set of adaptive testers contains all parallel and all classically adaptive testers. \\
    There exist testers which are parallel and not classically adaptive and vice versa, as well as testers which are parallel and classically adaptive. 
\end{Theorem}
Furthermore, we show that there exist standard channel discrimination tasks for which parallel strategies strictly outperform classically adaptive ones. In particular, when discriminating the square root of qubit Pauli channels with two copies, we could show that no classically adaptive tester can outperform the optimal parallel strategy. \new{To show this, we relax the optimisation by demanding that all tester elements $T_{\bf{IO}}^{i}$ have a positive partial transpose in the bipartition $\rm{I_1}\rm{O_1}|\rm{I_2}\rm{O_2}$, which is obviously follows from Eq.~\eqref{eq:two_copy_classically_adaptive_tester_def}.}

On the other hand, there exist standard channel discrimination tasks for which classically adaptive strategies strictly outperform parallel ones. \new{This fact may be demonstrated by a seesaw optimisation giving rise to lower bounds to the optimal classically adaptive strategy.} By computing bounds to the respective constraint separability problems, we obtained the following.
\begin{Example}
    \label{ex:ca_vs_par}
    Classically adaptive protocols outperform parallel ones in the discrimination of two copies of the uniform ensemble of an amplitude damping $\mathcal{C}_{\rm{ad}, \frac{2}{3}}$ a bit-flip channel $\mathcal{C}_{\rm{bf}, \frac{1}{3}}$, see Definition~\ref{def:important_channels}, and the identity $\textnormal{id}_{2}$. Here, the maximal success probability in the parallel case is computed as $p_{\rm{par}} = 0.80697$ while $p_{\rm{ca}} \geq 0.8118$ holds for the optimal classically adaptive schemes. \new{This lower bound has be obtained by a seesaw optimisation over the set of classically adaptive testers.}

    In return, for the discrimination of square roots of Pauli operators, one obtains $p_{\rm{par}} = 0.9571$ for the optimal parallel strategy and $p_{\rm{ca}} \leq 0.8980$ for the optimal classically adaptive strategy, \new{where the upper bound has been obtained by demanding that all tester elements have a positive partial transpose in the bipartition $\rm{I_1}\rm{O_1}|\rm{I_2}\rm{O_2}$.}
\end{Example}

 \new{The public repository \cite{gitcode} contains the script ``parallel\_vs\_class\_adp\_demonstration.jl'' that allows for a verification of the numerical results presented in Example \ref{ex:ca_vs_par}.}

This shows that there is no hierarchy between parallel and classically adaptive channel discrimination strategies, already in the consideration of standard channel discrimination tasks. For the computation of success probabilities in the case of classically adaptive testers, we again employ a formulation in terms of constrained separable states. 

\subsection{Optimisation over classically adaptive testers as constrained separability problem}
The collection of classically adaptive tester elements as in Definition~\ref{def:class_adp} can be described by a block matrix $\underline{T}_{\rm N\bf{IO}} := \bigoplus_{i=1}^{N} T_{\bf{IO}}^{i}$ 
that can be decomposed as 
\begin{equation}
    \underline{T}_{\rm N\bf{IO}} = P (\Pi [\underline{\underline{S}}_{\rm LNI_{2}O_{2}} \otimes  \underline{R}_{\rm L'I_{1}O_{1}}])P^{\dagger}
\end{equation}
with a permutation $P$ of the systems $\rm{N, I_1, O_1, I_2}$ and $\rm{O_2}$ and $\Pi$ is the same projection as in Eq.~\eqref{eq:con_sep_adpt_projection}.

The block matrices are defined as $\underline{\underline{S}}_{\rm LNI_{2}O_{2}}:=\bigoplus_{j=1}^{L} \bigoplus_{i=1}^{N} S_{\rm I_2 O_2}^{i|j}$ and $\underline{R}_{\rm L'I_{1}O_{1}}:= \bigoplus_{j=1}^{L} R_{\rm I_1 O_1}^{j}$. As $\underline{\underline{S}}_{\rm LNI_{2}O_{2}}$ and $\underline{R}_{\rm L'I_{1}O_{1}}$ are positive, $\underline{T}_{\rm N\bf{IO}}$ corresponds to a separable state. In particular, it is constrained separable with respect to the affine constraints 
\begin{align}
    &\Tr_{\rm L'}(\Omega_{\rm O_{1}}(\underline{R}_{\rm L'I_{1}O_{1}})) = 0\\
    &\Tr_{\rm LN}[(\ketbra{j}{j}_{\rm L} \otimes \mathds{1}_{\rm NI_{2}O_{2}}) \; \Omega_{\rm O_{2}}(\underline{\underline{S}}_{\rm LNI_{2}O_{2}})] = 0
\end{align}
for all $j = 1,\dots,L$ where $\Omega_{\rm O}$ is the ``trace-and-replace'' map defined as
\begin{equation}
    \Omega_{\rm O}(X_{\rm AO}) := X_{\rm AO} -  \Tr_{\rm O}(X_{\rm AO}) \otimes \frac{\mathds{1}_{\rm O}}{d_{\rm O}}.
\end{equation}

\begin{figure}
    \centering
    \includegraphics[width=\linewidth]{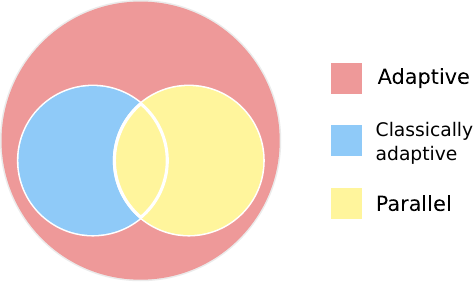}
    \caption{Venn diagram of the sets of different adaptive testers. Classically adaptive testers and parallel testers form intersecting subsets of the set of adaptive testers.} 
    \label{fig:adaptive sets}
\end{figure}

\section{Conclusion}
In this manuscript, we investigated the role of quantum memory in protocols for the task of channel discrimination. Our main contribution is the quantitative characterisation of quantum memory by constrained separability problems. In particular, we provided methods to compute optimal performances for arbitrary channel discrimination tasks in the memory restricted setting. Our analysis showed that the role of memory is subtle, especially in adaptive protocols, where a classical memory may sometimes suffice for an optimal discrimination success. Finally, we analysed classically adaptive schemes as a particular subset of the adaptive testers and explained their particular role in comparison to parallel schemes. \new{An interesting direction for future work is the employment of symmetry reductions in the SDP implementation, as discussed in Ref.~\cite{Vieira2024}, to further improve the scalability of our numerical techniques.} From a broader perspective, our results can be expected to be useful to determine the necessary memory for a quantum computation, or to design experiments to certify the dimension of a quantum memory. Interesting future topics are the study of memory limitations in deterministic quantum processes and in continuous parameter estimation tasks.

\section{Code availability}
\label{sec:code}
The source code used for our computations is available on a public repository \cite{gitcode}. \new{All the numerical values presented in Examples  \ref{ex:sep_vs_consep} and \ref{ex:ca_vs_par} and in Table \ref{tab:square_root} can be explicitly verified by execution of individual code scripts.}
\section*{Acknowledgements}

The authors would like to thank Otfried Gühne, Lucas Porto, Leonardo Vieira, and  Philip Taranto for helpful discussions.
The University of Siegen is kindly acknowledged for enabling our computation through the \texttt{OMNI} cluster. 
This work was supported by 
the Deutsche Forschungsgemeinschaft (DFG, German Research Foundation, project numbers 447948357 and 440958198), 
the Sino-German Center for Research 
Promotion (Project M-0294), 
the ERC (Consolidator Grant 683107/TempoQ), 
the German Ministry of Education 
and Research (Project QuKuK, BMBF 
Grant No. 16KIS1618K),
and the EIN~Quantum~NRW.

\clearpage

\onecolumn
\appendix

\section*{Appendix} 

\section{Realisation theorem for testers} 
\label{app:phys_impl}

In this appendix, we present a proof of Proposition \ref{prop:single_copy_phys_char} and Proposition \ref{prop:adpative_charcterisation} of the main text. An important tool for the proofs is the following identity. Let $\rm{A}$ and $\rm{B}$ be isomorphic quantum systems and let $\ket{X}_{\rm AB} = \sum_{i=1}^{d_{\rm A}} \ket{i} \otimes X \ket{i}$ where $X$ is some hermitian operator. 

Then, by definition of the link product it can be straightforwardly verified that the identity $\ketbra{X}_{\rm AB}*Y_{\rm BC}= \left(X_{\rm A}^T\otimes \mathds{1}_{\rm C}\right)\rho_{\rm AC} \left(X_{\rm A}^T\otimes \mathds{1}_{\rm C}\right)^\dagger$ holds. 

\begin{proposition*}
    Let $\{T_{\rm{IO}}^{i}\}_{i=1}^{N}$ be a collection of positive operators. Then, there exists a quantum state $\sigma_{I}$ such that $\sum_{i=1}^{N} T_{\rm{IO}}^{i} = \sigma_{\rm I} \otimes \mathds{1}_{\rm O}$ if and only if there exists a quantum state $\rho_{\rm IE}$ and a measurement $\{M_{\rm{EO}}^{i}\}_{i=1}^{N}$ such that $T_{\rm{IO}}^{i} = \rho_{\rm{IE}} * (M_{\rm{EO}}^{i})^{T}$
\end{proposition*}

\begin{proof}
    Assume first that $T_{\rm IO}^{i} = \rho_{\rm IE} * M_{\rm EO}^{i}$. One computes 
    \begin{align}
       \sum_{i=1}^{N} T_{\rm IO}^{i} &= \sum_{i=1}^{N} \rho_{\rm IE} * M_{\rm EO}^{i} \\
       &= \rho_{\rm IE} * \mathds{1}_{\rm EO} \\
       &= \Tr_{\rm E}(\rho_{\rm IE}) \otimes \mathds{1}_{\rm O}
    \end{align}
    which shows the first implication by defining $\sigma_{\rm I} := \Tr_{\rm E}(\rho_{\rm IE})$. 

    For the other direction, assume that $\sum_{i=1}^{N} T_{\rm IO}^{i} = \sigma_{\rm I} \otimes \mathds{1}_{\rm O}$ where $\sigma_{\rm I}$ is some quantum state. We choose the system $\rm{E}$ to be isomorphic to $\rm{I}$ and consider the pure state $\rho_{\rm IE} = \ketbra{\sqrt{\sigma}^{T}}_{\rm IE}$ via  
    \begin{equation}
        \ket{\sqrt{\sigma}^{T}}_{\rm IE} = \sum_{i=0}^{d_{I}-1} \ket{i} \otimes  \sqrt{\sigma}^{T} \ket{i}.
    \end{equation}
    Here the positive square root $\sqrt{X}$ is the operator $\sqrt{X}=\sum_i \sqrt{\alpha_i} \ketbra{\psi_i}$ where $X = \sum_i \alpha_{i} \ketbra{\psi_i}$ is a spectral decomposition of a hermitian operator $X$.
    The state $\rho_{\rm IE}$ is a purification of $\sigma_{\rm I}^{T}$ in the sense that $\Tr_{\rm E}(\rho_{\rm IE}) = \sigma_{\rm I}^{T}$. Next, we define a measurement $M_{\rm EO}^{i}$ by 
    \begin{equation}
        M_{\rm EO}^{i} = (\sqrt{\sigma}^{-1}_{\rm E} \otimes \mathds{1}_{\rm O}) \; T_{\rm EO}^{i} \; (\sqrt{\sigma}^{-1}_{\rm E} \otimes \mathds{1}_{\rm O}) + G_{\rm EO}^{i}
    \end{equation}
    where ${\sigma}^{-1}=\sum_{\alpha_i\neq0} \frac{1}{\alpha_i} \ketbra{\psi_i}$ denotes the Moore–Penrose inverse (inverse on $\textnormal{im}(\sigma)$) and where $G_{\rm EO}^{i}$ is an arbitrary collection of positive operators that sums up to the projection on  $\textnormal{im}(\sigma)^{\perp}$, i.e.,  $\sum_{i} G_{\rm EO}^{i} = \Pi_{\textnormal{im}(\sigma)^{\perp}}$.
    Then it can be simply calculated that $T_{\rm IO}^i = \rho_{\rm IE} * (M_{\rm EO}^{i})^{T}$ completing the proof. 
\end{proof}

We now present the proof of Proposition \ref{prop:adpative_charcterisation} from the main text.
\begin{proposition*}
    A collection of positive operators $\{T_{\bf{IO}}^{i}\}_{i=1}^{N}$  with $W_{\bf{IO}} := \sum_{i=1}^{N} T_{\bf{IO}}^{i}$ is a two-copy adaptive tester if and only if 
    \begin{align}
    W &= R_{\rm \mathbf{I} O_1} \otimes \mathds{1}_{\rm O_{2}} \\
    \Tr_{\rm I_{2}}(R_{\rm \mathbf{I} O_1}) &= \sigma_{\rm I_{1}} \otimes \mathds{1}_{\rm O_1} 
\end{align}
where $R_{\rm \mathbf{I} O_1}$ is a positive operator and $\sigma_{\rm I_{1}}$ is a state.
\end{proposition*} 

\begin{proof}
Similarly to the single-copy case, we start by setting the system $\rm{E}_1$ to be isomorphic to $\textnormal{im}(\sigma_{\rm I_{1}})$ and consider the pure state $\rho_{\rm I_1E_1} = \ketbra{\sqrt{\sigma}^{T}}_{\rm I_1E_1}$ via  
\begin{equation}
        \ket{\sqrt{\sigma}^{T}}_{\rm I_1E_1} = \sum_{i=0}^{d_{\rm I}-1} \ket{i} \otimes  \sqrt{\sigma}^{T} \ket{i}.
    \end{equation}
We now set $\rm{E}_2$ to be a system isomorphic to $\rm{E}_1 \rm{O}_1 \rm{I}_2 $ and define a quantum channel 
$\mathcal{K}_{\rm E_{1} O_{1} \rightarrow I_{2} E_{2}}$ via the Choi operator
\begin{equation}
    K_{\rm E_{1}O_{1}I_{2}E_{2}} = (\sqrt{\sigma}^{-1}_{\rm E_{1}} \otimes \mathds{1}_{\rm O_{1}I_{2}E_{2}}) \ketbra{\sqrt{R}}_{\rm E_{1}O_{1}I_{2}E_{2}} (\sqrt{\sigma}^{-1}_{\rm E_{1}} \otimes \mathds{1}_{\rm O_{1}I_{2}E_{2}})
\end{equation}
The operator $K_{\rm E_{1}O_{1}I_{2}E_{2}}$ is positive semidefinite and a direct calculation shows that $\Tr_{\rm I_2E_2}(K) = \mathds{1}_{\rm E_1O_1}$, hence $K_{\rm E_{1}O_{1}I_{2}E_{2}}$ is indeed the Choi operator of a quantum channel $\mathcal{K}_{\rm E_{1} O_{1} \rightarrow I_{2} E_{2}}$.

Finally, we define the POVM $\{M_{\rm E_2O_2}^{i}\}_{i=1}^{N}$ via
    \begin{equation}
        M_{\rm E_2O_2}^{i} := (\sqrt{R}^{-1}_{E_2} \otimes \mathds{1}_{\rm O_2}) \; (T_{\rm E_2O_2}^{i})^{T} \; (\sqrt{R}^{-1}_{\rm E_2} \otimes \mathds{1}_{\rm O_2}) + G_{\rm E_2O_2}^{i},
    \end{equation}
     where $G_{\rm E_2O_2}^{i}$ is an arbitrary collection of positive operators that sums up to the projection on  $\textnormal{im}(R)^{\perp} \otimes \mathcal{H}_{\rm O_{2}}$, i.e.,  $\sum_{i} G_{\rm EO}^{i} = \Pi_{\textnormal{im}(R)^{\perp}} \otimes \mathds{1}_{\rm O_{2}}$. Then it can be checked that $T_{\rm I_1O_1I_2O_2}^i = \rho_{\rm I_1E_2} * K_{\rm E_1O_1I_2E_2} * (M_{\rm E_2O_2}^{i})^{T}$  completing the proof.
\end{proof}

\section{Proof of Proposition \ref{prop:clock_shift_memory_dependence}}
\label{prop:heisenberg_weyl}
Here, we present the proof of Theorem~\ref{prop:clock_shift_memory_dependence} in the main text that stated the following.
\begin{Theorem*}
    The maximal success probability $p_{d, d_{\rm E}}$ in discrimination of the uniform ensemble given by the $d$-dimensional clock-shift operators 
    by memory-$d_{E}$ single-copy testers is given by 
    \begin{equation}
        p_{d, d_{\rm E}} = \min \left\{1,  \frac{d_{\rm E}}{d} \right\}.
    \end{equation}
\end{Theorem*}
\begin{proof}
    To start, in  \cite{Chiribella2005} it has been shown that if $\{p_i,U^i\}_{i=1}^N$ is an ensemble of unitary operations, $p_i=1/N$ (uniform distribution) and the unitary operators $U^i\in\mathcal{L}(\mathbb{C}^d)$ form a group up to a global phase (i.e., $U^i$ is a projective representation of a group), the maximal probability for discriminating the $N$ unitary from the ensemble $\{p_i,U^i\}_{i=1}^N$ in a single shot memory-less scenario is given by
    \begin{align}
        \label{eq:irrep_suc_prob}
        p_s=\frac{1}{N}\sum_{\lambda\in\text{irreps}} d_\lambda \min{\left\{d_\lambda, m_\lambda\right\}},
    \end{align}
    where $\lambda$ labels the irreducible representations (irreps) of the group formed by $\{U^i\}_{i=1}^N$ represented in $\mathbb{C}^d$ and $d_{\lambda}$ and $m_{\lambda}$ are the corresponding dimensions and multiplicities, respectively.

    Now, we notice that the clock shift operators $\{X_d^iZ_d^j\}_{i,j=0}^{d-1}$ form a group up to a global phase. Also, this group cannot be reduced on $\mathbb{C}^d$. To see that it is irreducible, notice that the character of the representation is given by $\chi(X_d^i Z_d^j)=\tr(X_d^i Z_d^j)$, and for $i=j=0$ we have $\chi(I_d)=d$; in all other cases we have $\chi(X_d^i Z_d^j)=0$, because all clock-shift operators that are not the identity are traceless. The Schur orthogonality relations~\cite[Theorem~5]{Serre1977} state that a representation of a finite group $G$ is irreducible if and only if the characters of it respect
    \begin{align}
        \sum_{g\in G} \frac{1}{\abs{G}} \chi(g)\overline{\chi(g)} =1.
    \end{align}
    For the particular case of the clock and shift operators it follows that
    \begin{align}
        \sum_{i,j\in\{0,\ldots,d-1\}} \frac{1}{d^2} \tr(X_d^i Z_d^j)\overline{\tr(X_d^i Z_d^j)} = 1,
    \end{align}
    hence the representation of the clock-shift operators in $\mathbb{C}^d$ is irreducible, and since it has dimension $d$, the multiplicity of this irrep is one. This is enough to show that in the memory-less case, we have that
    \begin{align}
        p_s=&\frac{1}{d^2} d \\
        =&\frac{1}{d}.
    \end{align}

    Now, we recall Remark \ref{obs:memory-less-reduction} that having a memory size of $d_{\rm E}$ is equivalent of discriminating the ensemble with operators $\{X_d^iZ_d^j\otimes I_{d_{\rm E}}\}_{i,j=0}^{d-1}$. We notice that the set $\{X_d^iZ_d^j\otimes I_{d_{\rm E}}\}_{i,j=0}^{d-1}$ is itself a projective representation of the clock-shift operators, with a single irrep of dimension $d$, and multiplicity $d_{\rm E}$. Hence, evaluating Eq.~\eqref{eq:irrep_suc_prob} yields
    \begin{align}
        p_s=&\frac{1}{d^2} d \min{\{d,d_{\rm E}\}} \\
        =&\frac{\min{\{d,d_{\rm E}\}}}{d} = \min \left\{1,  \frac{d_{{\rm E}}}{d} \right\}.
    \end{align}
\end{proof}

\section{Quantum memory and constrained separability problems}
\label{sec:con_sep}
\subsection{Constrained separable states}
Before we consider constrained separability problems as well as their ability to quantify quantum memory, let us recall the standard notion of quantum separability. A bipartite (mixed) quantum state $\rho_{\rm AB}$ is separable if it can be decomposed as 
\begin{equation} \label{eq:sep_deco}
    \rho_{\rm AB} = \sum_{\lambda} p_{\lambda} \; \rho_{\rm A}^{\lambda} \otimes \rho_{\rm B}^{\lambda}
\end{equation}
with a probability distribution $\{p_{\lambda}\}_{\lambda}$ and states $\{\rho_{\rm A}^{\lambda}\}_{\lambda}$ and $\{\rho_{\rm B}^{\lambda}\}_{\lambda}$ of the systems $\rm{A}$ and $\rm{B}$, respectively. Constrained separable states introduced in Ref.~\cite{Berta2021} obey a slightly modified definition compared to the usual notion of separability.
\begin{definition}[Constrained separable state] \label{def:con_sep}
    Let $\rho_{\rm AB}$ be a bipartite quantum state, $\Phi_{\rm A}: \mathcal{L}(\mathcal{H}_{A}) \rightarrow V_{\rm A}$ and $\Psi_{\rm B}: \mathcal{L}(\mathcal{H}_{\rm B}) \rightarrow V_{\rm B}$ be linear maps to finite-dimensional vector spaces $V_{\rm A}$ and $V_{\rm B}$ and let $a \in V_{\rm A}$ and $b \in V_{\rm B}$ be some constant vectors. Then, $\rho_{\rm AB}$ is constrained separable with respect to $\Phi_{\rm A}$, $\Psi_{\rm B}$, $a$ and $b$ if $\rho_{\rm AB}$ can be decomposed as 
    \begin{equation}
    \rho_{\rm AB} = \sum_{\lambda} p_{\lambda} \; \rho_{\rm A}^{\lambda} \otimes \rho_{\rm B}^{\lambda}
\end{equation}
where $\{p_{\lambda}\}_{\lambda}$ is a probability distribution, $\rho_{\rm A}^{\lambda} \geq 0$, $\rho_{\rm B}^{\lambda} \geq 0$, $\textnormal{Tr}(\rho_{\rm A}^{\lambda}) = \textnormal{Tr}(\rho_{\rm B}^{\lambda}) = 1$  and 
\begin{align}
    \Phi_{\rm A}(\rho_{\rm A}^{\lambda}) = a \\
    \Psi_{\rm B}(\rho_{\rm B}^{\lambda}) = b 
\end{align}
for all $\lambda$. 
\end{definition}
In words, constrained separable states are separable states for which the tensor factors in each summand of each separable decomposition obey an additional affine constraint.
When it follows from the context which constraining maps and constants are under consideration, we will just speak of constrained separable states without further reference.  

\subsection{Optimisation over constrained separable states}
An optimisation of a linear functional over some set of constrained separable states is a constrained separability problem \cite{Berta2021}. Every bipartite constrained separability problem is characterised by the data tuple $(F_{\rm AB}, \Phi_{\rm A}, \Psi_{\rm B}, a, b)$ and can be formulated as 
\begin{align}
\label{eq:csep_problem}
    \max_{\rho_{\rm AB}} &\Tr(F_{\rm AB} \, \rho_{\rm AB}) \\
    \textnormal{s.t.} \; &\rho_{\rm AB} = \sum_{\lambda} p_{\lambda} \; \rho_{\rm A}^{\lambda} \otimes \rho_{\rm B}^{\lambda} \\
    & \rho_{\rm A}^{\lambda} \geq 0, \; \Tr(\rho_{\rm A}^{\lambda}) = 1, \; \Phi_{\rm A}(\rho_{\rm A}^{\lambda}) = a \textnormal{ for all } \lambda \\
    & \rho_{\rm B}^{\lambda} \geq 0, \; \Tr(\rho_{\rm B}^{\lambda}) = 1, \; \Psi_{\rm B}(\rho_{\rm B}^{\lambda}) = b \textnormal{ for all } \lambda \\
    & p_{\lambda} \geq 0, \; \sum_{\lambda} p_{\lambda} = 1. 
\end{align}
Our strategy to find bounds or even exact results in some cases consists of calculating lower and upper bounds to the problem \eqref{eq:csep_problem} by semidefinite programs. If we find that lower and upper bounds coincide for a problem, we can conclude that we have found the optimum.

\subsection{Seesaw optimisation}
\label{sec:seesaw}
Lower bounds to the constrained separability problem \eqref{eq:csep_problem} can be obtained by a seesaw algorithm, that is a standard approach in bilinear or more generally biconvex optimisation problems \cite{Gorski2007}. By the linearity of the objective function, it is clear that the optimum will be obtained on an extreme point $\rho_{\rm A} \otimes \rho_{\rm B}$ of the set of constrained separable states. Furthermore, every evaluation $\Tr(F_{\rm AB} \; \rho_{\rm A} \otimes \rho_{\rm B})$ for fixed $\rho_{\rm A}, \rho_{\rm B}$ is a lower bound to the optimum of \eqref{eq:csep_problem}. The idea is now to fix a local state $\rho_{\rm A}$ that fulfils $\Phi_{\rm A}(\rho_{\rm A}) = a$ and to maximise the reduced linear functional over states $\rho_{\rm B}$ that fulfil $\Psi_{\rm B}(\rho_{\rm B}) = b$. This optimisation problem is a semidefinite program in standard form. In the next step, one takes the optimiser $\rho_{\rm B}$ and performs the analogous optimisation over $\rho_{\rm A}$. Both optimisation problems that are given by 
\begin{equation}
    \begin{split}
        &\textnormal{given} \; \rho_{\rm A}  \\
        &\max_{\rho_{\rm B}} \; \Tr[\Tr_{\rm A}(F_{\rm AB} (\rho_{\rm A} \otimes \mathds{1}_{\rm B})) \rho_{\rm B}] \\
        & \textnormal{s.t.} \; \rho_{\rm B} \geq 0, \Tr(\rho_{\rm B}) = 1, \Psi_{\rm B}(\rho_{\rm B}) = b \\
        &\Rightarrow \textnormal{Evaluate } \rho_{\rm B}
    \end{split}
    \quad\quad\quad\quad
    \begin{split}
        &\textnormal{given} \; \rho_{\rm B}  \\
        &\max_{\rho_{\rm A}} \; \Tr[\Tr_{\rm B}(F_{\rm AB} (\mathds{1}_{\rm A} \otimes \rho_{\rm B})) \rho_{\rm A}] \\
        & \textnormal{s.t.} \; \rho_{\rm A} \geq 0, \Tr(\rho_{\rm A}) = 1, \Phi_{\rm A}(\rho_{\rm A}) = a \\
        &\Rightarrow \textnormal{Evaluate } \rho_{\rm A}
    \end{split}
\end{equation}
are carried out alternately until convergence occurs. The limit point that is achieved by this procedure is \textbf{not} guaranteed to be optimal value. However, by sampling over many initialisations $\rho_{\rm A}$, one can often find a promising candidate for the optimiser and the optimum. Often, upper bounds can then be used to confirm the optimality.

But even in the absence of any upper bounds, the error in the approximation coming from the seesaw algorithm can be explicitly bounded, based on the geometry of the polytope that is spanned by the chosen initialisations. To make this precise, we need to define a few quantities. First, the setting of the constrained separability problems naturally motivates the definition of the associated constrained (local) state spaces $Y_{\rm A}$ given by the convex set of states of the form
\begin{equation}
     Y_{\rm A} := \{\rho_{\rm A} \geq 0, \Tr(\rho_{\rm A}) = 1, \Phi_{\rm A}(\rho_{\rm A}) = a\}.
\end{equation}
As mentioned above, in seesaw optimisation one selects or randomly samples states from $Y_{\rm A}$. The initialisations may be thought of points that span a polytope that is included in $Y_{\rm A}$ motivating the following definition of inner and outer polytopes of a constrained state space.   

\begin{definition}[Inner and outer polytope of constrained state space] 
    Let $Y_{\rm A}$ be a constrained state space. An inner polytope $\mathcal{V}_{\rm A}$ of $Y_{\rm A}$ is the convex hull of a finite subset of elements of $Y_{\rm A}$, i.e., $\mathcal{V}_{\rm A} = \textnormal{conv}(\{\widetilde{\rho}_{\rm A}^{\lambda}\}_{\lambda = 1}^{m})$ where for each $\lambda \in \{1,\dots,m\}$,  $\widetilde{\rho}_{\rm A}^{\lambda} \in Y_{\rm A}$. An outer polytope of a constrained state space $Y_{\rm A}$ is the convex hull $\mathcal{O}_{\rm A} = \textnormal{conv}(\{\sigma_{\rm A}^{\lambda}\}_{\lambda = 1}^{m})$ of hermitian operators $\sigma_{\rm A}^{\lambda} \in \mathcal{L}(\mathcal{H}_{\rm A})$, such that $\mathcal{O}_{\rm A}$ contains $Y_{\rm A}$, i.e., for each $\rho_{\rm A} \in Y_{\rm A}$ there is a probability distribution $\{q_{\lambda}\}_{\lambda=1}^{m}$ such that $\rho_{\rm A} = \sum_{\lambda = 1}^{m} q_{\lambda} \, \sigma_{\rm A}^{\lambda}$.
\end{definition}

Inner polytopes $\mathcal{V}_{\rm A}$ that lie in the constrained state space and outer polytopes $\mathcal{O}_{\rm A}$ that contain the state space give rise to lower and upper bounds $r_{\mathcal{V}_{\rm  A}}$ and $r_{\mathcal{O}_{\rm  A}}$ of the optimisation problem \eqref{eq:csep_problem} via

\begin{equation}
    \begin{split}
        r_{\mathcal{V}_{\rm  A}} = \max_{\lambda \in \{1,\dots,m\}} &\max_{\rho_{\rm B}} \; \Tr[\Tr_{\rm A}(F_{\rm AB} (\widetilde{\rho}_{\rm A}^{\lambda} \otimes \mathds{1}_{\rm B})) \rho_{\rm B}] \\
        & \textnormal{s.t.} \; \rho_{\rm B} \geq 0, \Tr(\rho_{\rm B}) = 1 \\
        &\Psi_{\rm B}(\rho_{\rm B}) = b 
    \end{split}
    \quad
    \begin{split}
        r_{\mathcal{O}_{\rm  A}} = \max_{\lambda \in \{1,\dots,m\}} &\max_{\rho_{\rm B}} \; \Tr[\Tr_{\rm A}(F_{\rm AB} (\sigma_{\rm A}^{\lambda} \otimes \mathds{1}_{\rm B})) \rho_{\rm B}] \\
        & \textnormal{s.t.} \; \rho_{\rm B} \geq 0, \Tr(\rho_{\rm B}) = 1 \\
        &\Psi_{\rm B}(\rho_{\rm B}) = b 
    \end{split}
\end{equation}
that can both be evaluated by finitely many semidefinite programs.

Inner and outer polytopes are therefore practical approximations of the constrained state space $Y_{\rm A}$. To quantify the quality of the approximation of inner polytopes, whose vertices correspond to a potential set of the initialisations of the seesaw optimisation, we define the approximation radius.

\begin{definition}[Approximation radius]
    Let $\mathcal{V}_{\rm A}$ be an inner polytope of a constrained state space $Y_{\rm A}$. Furthermore, let $\tau_{A}$ be some state that is in the interior of $\mathcal{V}_{\rm A}$. Then, the approximation radius $l_{\tau_{\rm A}}(\mathcal{V}_{\rm A})$ of $\mathcal{V}_{\rm A}$ with respect to the reference state $\tau_{\rm A}$ is defined as the solution of the optimisation problem 
    \begin{align}
        l_{\tau_{\rm A}}(\mathcal{V}_{\rm A}) := &\max_{t}  t \\
    &\textnormal{s.t} \; t \rho_{\rm A} + (1-t) \tau_{\rm A} \in \mathcal{V}_{\rm A} \; \forall \rho_{\rm A} \in Y_{\rm A}.
    \end{align}
    
\end{definition}

\new{
One way to compute the approximation radius in practice works by first enumerating the facets of the polytope $\mathcal{V}_{\rm A}$. For that, one has to find hermitian linear operators $B_{\rm A}^{i} \in \mathcal{L}(\mathcal{H}_{\rm A})$ and numbers $b_i \in \mathbb{R}$ for $i \in \{1,\dots,L\}$ such that 
\begin{equation}
    \mathcal{V}_{\rm A} = \bigcap_{i=1}^{L}\{\rho_{\rm A}: \Tr(B_{\rm A}^{i} \rho_{\rm A}) \leq b_i\}.
\end{equation}
Given the vertices of the polytope, the facet enumeration may be performed using the algorithm of the double description method \cite{Motzkin1953}. The approximation radius can then be computed by the minimisation over finitely many outcomes of semidefinite programs, one for each facet of the polytope, as follows.  
\begin{align}
        l_{\tau_{\rm A}}(\mathcal{V}_{\rm A}) = \, \underset{i \in \{1,\dots,L\}}{\min} \; &\min \; s \label{eq:app_radius_sdp}\\
    & \textnormal{w.r.t} \; s \in \mathbb{R}_{+}, \; \widetilde{\rho}_{\rm A}  \in \mathcal{L}(\mathcal{H}_{\rm A})  \nonumber\\
    &\textnormal{s.t} \; \Tr(B_{\rm A}^{i} \widetilde{\rho}_{\rm A}) \geq b_{i} - (1-s) \Tr(B_{\rm A}^{i} \tau_{\rm A}) \nonumber\\
    &\quad \;\widetilde{\rho}_{\rm A} \geq 0, \; \Tr(\widetilde{\rho}_{\rm A}) = s, \, \Phi_{\rm A}(\widetilde{\rho}_{\rm A}) = sa \nonumber.
\end{align}
}

The error bound in Theorem \ref{thm:see_saw_bound} of the main text further involves the quantity $f_{\tau_{\rm A}}$ defined as 

\begin{equation}
    f_{\tau_{\rm A}} = \min_{\rho_{\rm B}\in Y_{\rm B}} \Tr \left[\Tr_{\rm A}(F_{\rm AB} \, \tau_{\rm A} \otimes \mathds{1}_{\rm B}) \rho_{\rm B}\right]
\end{equation}
where the minimisation runs over all states $\rho_{\rm B}$ in the constrained state space $Y_{\rm B}$.
Importantly, $f_{\tau_{\rm A}}$ only depends on the reference state $\tau_{\rm A}$ and can be computed by a single SDP. We are now ready to prove Theorem \ref{thm:see_saw_bound}.

\begin{Theorem*}
    Let $(F_{\rm AB}, \Phi_{\rm A}, a, \Psi_{\rm B}, b)$ describe a constrained separability problem and let $\mathcal{V}_{\rm A} = \textnormal{conv}(\{\widetilde{\rho}_{\rm A}^{\lambda}\}_{\lambda=1}^{m})$ be an inner polytope of the constrained state space $Y_{\rm A}$, that has a non-zero approximation radius $l_{\tau_{\rm A}}(\mathcal{V}_{\rm A})$ with respect to some reference state $\tau_{\rm A}$. Then, for the polytope approximation $r_{\mathcal{V}_{\rm A}}$, we have the inequalities  
    \begin{equation}
        r_{\mathcal{V}_{\rm A}} \leq r_{\rm opt} \leq \frac{1}{l_{\tau_{\rm A}}(\mathcal{V}_{\rm A})} r_{\mathcal{V}_{\rm A}} + \frac{l_{\tau_{\rm A}}(\mathcal{V}_{\rm A})-1}{l_{\tau_{\rm A}}(\mathcal{V}_{\rm A})} f_{\tau_{\rm A}}
    \end{equation}
    In particular, the right-hand side of the inequality goes to $r_{\mathcal{V}_{\rm A}}$ if the approximation radius $l_{\tau_{\rm A}}(\mathcal{V}_{\rm A})$ goes to $1$.
\end{Theorem*}
\begin{proof}
    First for $\lambda = 1,\dots, m$, let us consider the operators
    $\sigma_{\rm A}^{\lambda}$ defined via 
    \begin{equation}
        \sigma_{\rm A}^{\lambda} := \frac{1}{l_{\tau_{\rm A}}(\mathcal{V}_{\rm A})}(\widetilde{\rho}_{\rm A}^{\lambda} - \tau_{\rm A}) + \tau_{\rm A}. 
    \end{equation}
    We now claim that $\mathcal{O}_{\rm A} = \textnormal{conv}(\{\sigma_{\rm A}^{\lambda}\}_{\lambda=1}^{m})$ is an outer polytope of $Y_{\rm A}$. To show this, let $\rho_{\rm A} \in Y_{\rm A}$. By definition of the approximation radius $l_{\tau_{\rm A}}(\mathcal{V}_{\rm A})$, we know that the state $\rho_{l, \rm{A}}$
    given by 
    \begin{equation}
    \label{eq:rho_lA}
        \rho_{l, \rm{A}} := l_{\tau_{\rm A}}(\mathcal{V}_{\rm A})\rho_{\rm A} + (1-l_{\tau_{\rm A}}(\mathcal{V}_{\rm A}))\tau_{\rm A}
    \end{equation}
    lies in the polytope $\mathcal{V}_{\rm A}$ which means that there is a decomposition $\rho_{l, \rm{A}} = \sum_{\lambda = 1}^{m} q_{\lambda} \, \widetilde{\rho}_{\rm A}^{\lambda}$ with some probability distribution $\{p_{\lambda}\}_{\lambda=1}^{m}$. Reordering Eq.~\eqref{eq:rho_lA} yields
    \begin{align}
        \rho_{\rm A} &= \frac{1}{l_{\tau_{\rm A}}(\mathcal{V}_{\rm A})} \rho_{l, \rm{A}} + \frac{l_{\tau_{\rm A}}(\mathcal{V}_{\rm A})-1}{l_{\tau_{\rm A}}(\mathcal{V}_{\rm A})} \tau_{\rm A} \\
        &= \sum_{\lambda = 1}^{m} q_{\lambda} \left(\frac{1}{l_{\tau_{\rm A}}(\mathcal{V}_{\rm A})}\widetilde{\rho}_{\rm A}^{\lambda} + \frac{l_{\tau_{\rm A}}(\mathcal{V}_{\rm A})-1}{l_{\tau_{\rm A}}(\mathcal{V}_{\rm A})} \tau_{\rm A}\right) \\
        &= \sum_{\lambda = 1}^{m} q_{\lambda} \, \sigma_{\rm A}^{\lambda}
    \end{align}
    from which follows that $\rho_{\rm A} \in \mathcal{O}_{\rm A}$,  implying that $\mathcal{O}_{\rm A}$ is an outer polytope. Hence, we obtain the bound
    \begin{align}
        r_{\rm opt} &\leq \max_{\lambda} \max_{\rho_{\rm B}} \Tr\left[\Tr_{\rm A}(F_{\rm AB} \, \sigma_{\rm A}^{\lambda} \otimes  \mathds{1}_{\rm B}) \rho_{\rm B}\right] \\
        &\leq \frac{1}{l_{\tau_{\rm A}}(\mathcal{V}_{\rm A})} \max_{\lambda} \max_{\rho_{\rm B}}\left( \Tr\left[\Tr_{\rm A}(F_{\rm AB} \, \widetilde{\rho}_{\rm A}^{\lambda}  \otimes  \mathds{1}_{\rm B}) \rho_{\rm B}\right] \right) + \frac{l_{\tau_{\rm A}}(\mathcal{V}_{\rm A})-1}{l_{\tau_{\rm A}}(\mathcal{V}_{\rm A})} \min_{\rho_{\rm B}} \left(\Tr\left[\Tr_{\rm A}(F_{\rm AB} \, \tau_{\rm A}  \otimes  \mathds{1}_{\rm B}) \rho_{\rm B}\right]\right) \label{eq:second_ineqaulity_seesaw_proof} \\
        &= \frac{1}{l_{\tau_{\rm A}}(\mathcal{V}_{\rm A})} r_{\mathcal{V}} + \frac{l_{\tau_{\rm A}}(\mathcal{V}_{\rm A})-1}{l_{\tau_{\rm A}}(\mathcal{V}_{\rm A})} f_{\tau_{\rm A}}.
    \end{align}
    The minimisation in the second term of \eqref{eq:second_ineqaulity_seesaw_proof} arises from the fact that the prefactor $\frac{l_{\tau_{\rm A}}(\mathcal{V}_{\rm A})-1}{l_{\tau_{\rm A}}}$ is always negative. 
\end{proof}

\subsection{Outer approximations}
\subsubsection{PPT-criterion}
\label{sec:ppt_criterion}
A simple and important tool to compute upper bounds to constrained separability problems is the well known PPT criterion \cite{PhysRevLett.77.1413, Horodecki1996}. Every bipartite separable state $\rho_{\rm AB}$ has a positive partial transpose (PPT) $\rho_{\rm AB}^{T_{\rm B}}$. Hence, the set of positive operators with a positive partial transpose forms an outer approximation of the separable operators and states which are not PPT are detected as entangled. In general, this approximation is not optimal in the sense that there exist entangled states which are PPT. 

This criterion for standard separability can be easily generalised for constrained separable states. One can state the following.
\begin{proposition}
    Let the state $\rho_{\rm AB}$ be constrained separable with respect to $(\Phi_{\rm A}, \Psi_{\rm B}, a, b)$. Then it holds that
    \begin{align}
        \rho_{\rm AB}^{T_{\rm B}} &\geq 0 \label{eq:con_ppt_a} \\
        (\textnormal{id}_{\rm A} \otimes \Psi_{\rm B})[\rho_{\rm AB}] &= \Tr_{\rm B}(\rho_{\rm AB}) \otimes b \label{eq:con_ppt_b}\\
        (\Phi_{\rm A} \otimes \textnormal{id}_{\rm B})[\rho_{\rm AB}] &= a \otimes \Tr_{\rm A}(\rho_{\rm AB}). \label{eq:con_ppt_c}
    \end{align}
\end{proposition}
States $\rho_{\rm AB}$ that fulfil the conditions \eqref{eq:con_ppt_a}-\eqref{eq:con_ppt_c} hence form an outer approximation of the constrained separable states. Relaxing the constrained separability constraints to PPT enables an upper bound of the optimisation problem \eqref{eq:csep_problem}.

\subsubsection{Constrained symmetric extensions}
\label{sec:constrained_sym_extension}
A hierarchy of outer approximations of the set of constrained separable states is provided by so called $k-$constrained symmetric extendible states which generalise symmetric extendible states introduced in \cite{PhysRevA.69.022308}. The corresponding idea stems from the quantum de Finetti theorem that characterises exchangeable sequences of quantum states \cite{Caves2002, Christandl2007, PRXQuantum.3.010340}. The following characterisation of bipartite separable states is known. 
\begin{Theorem}[\cite{PhysRevA.69.022308, Christandl2007}]
\label{thm:dps_standard}
    Let $\rho_{\rm AB}$ be a bipartite quantum state such that there exists a sequence of states $\rho_{\textnormal{Sym}(\rm{A},k)\rm{B}}$ labelled by $k\in \mathbb{N}$ such that $\Tr_{\rm{A}^{k-1}}(\rho_{\textnormal{Sym}(\rm{A},k)\rm{B}}) = \rho_{\rm AB}$. Then every member of this sequence can be written as
    \begin{equation}
        \rho_{\textnormal{Sym}(\rm{A},k)\rm{B}} = \sum_{\lambda} p_{\lambda}^{(k)} \, {\rho_{\rm A}^{\lambda}}^{\otimes k} \otimes \rho_{\rm B}^{\lambda} 
    \end{equation}
    for some probability distribution $p_{\lambda}^{(k)}$ and quantum states $\rho_{A}^{\lambda}$ and $\rho_{B}^{\lambda}$. In particular, $\rho_{AB}$ is separable.
\end{Theorem}
The states $\rho_{\textnormal{Sym}(\rm{A},k)\rm{B}}$ are called symmetric extensions of $\rho_{\rm AB}$ of order $k$ and, if such an extension exists for some $k$, $\rho_{\rm AB}$ is said to be $k$-symmetric extendible. As such, the sets of $k$-symmetric extendible states for every $k \in \mathbb{N}$ serve as a converging hierarchy of outer approximations for the separable states.  
This concept admits a generalisation to constrained separable states in the following sense.

\begin{definition}[$k$-constrained symmetric extendible state]
    Let $\rho_{\rm AB}$ be a bipartite quantum state, $\Phi_{\rm A}: \mathcal{L}(\mathcal{H}_{\rm A}) \rightarrow V_{\rm A}$ and $\Psi_{\rm B}: \mathcal{L}(\mathcal{H}_{\rm B}) \rightarrow V_{\rm B}$ be linear maps, $V_{\rm A}$ and $V_{\rm B}$ be vector spaces, and let $a \in V_{\rm A}$ and $b \in V_{\rm B}$ be some constant vectors. Then $\rho_{\rm AB}$ is $k-$constrained symmetric extendible for $k \in \mathbb{N}$ with respect to $\Phi_{\rm A}$, $\Psi_{\rm B}$, $a$ and $b$ if there exists a quantum state $\rho_{\rm{A}^{k}\rm{B}}$ such that
    \begin{align} 
              \rho_{\rm AB} &= \Tr_{\rm{A}^{k-1}}(\rho_{\rm{A}^{k}\rm{B}}), \; \Tr(\rho_{\rm{A}^{k}\rm{B}})=1 \label{eq:extension} \\
              \rho_{\rm{A}^{k}\rm{B}} &= (U_{\sigma} \otimes \mathds{1}_{\rm B}) \rho_{\rm{A}^{k}\rm{B}} (U_{\sigma}^{\dagger} \otimes\mathds{1}_{\rm B}) \;\; \forall \sigma \in \mathcal{S}_{k} \label{eq:symmetry}\\
             \Tr_{\rm B}(\rho_{\rm{A}^k \rm{B}}) \otimes b &= (\textnormal{id}_{\rm{A}^{k}} \otimes \Psi_{\rm B})\rho_{\rm{A}^{k}\rm B} \label{eq:B_constrained}\\
             \Tr_{\rm AB}(\rho_{\rm{A}^{k}\rm{B}}) \otimes a &= (\textnormal{id}_{\rm{A}^{k-1}} \otimes \Phi_{\rm A}) \Tr_{\rm B}(\rho_{\rm{A}^k \rm{B}}) \label{eq:A_constrained}
    \end{align}
    where $\mathcal{S}_{k}$ is the symmetric group of $k$ elements and $\sigma \mapsto U_{\sigma}$ is its standard unitary representation on $(\mathbb{C}^{d_{\rm A}})^{\otimes k}$. 
\end{definition}
The conditions for $k$-constrained extendibility are semidefinite representable and the limit of $k \rightarrow \infty$ corresponds to the set of constrained separable states, which is formalised by the next theorem, first proven in Ref.~\cite{Berta2021} and generalised to arbitrary ordered vector spaces in \cite{aubrun2022monogamyentanglementcones}.
\begin{Theorem}[\cite{Berta2021}]
\label{thm:constrained_dps_convergence}
    A bipartite quantum state $\rho_{\rm AB}$ is constrained separable with respect to $\Phi_{\rm A}$, $\Psi_{\rm B}$, $a$ and $b$ if and only if it is $k-$constrained symmetric extendible for all $k \in \mathbb{N}$. 
\end{Theorem}
Similar to the unconstrained case, the convergence speed of the hierarchy can be improved by imposing additional PPT constraints on the extension $\rho_{\rm{A}^{k}\rm{B}}$. This works by adding the additional constraint
\begin{equation}
    (\rho_{\rm{A}^{k}\rm{B}})^{T_{\rm{A}^{l}}} \geq 0 \textnormal{ for all } l=1,\dots,k
\end{equation}
to the constraints \eqref{eq:extension}-\eqref{eq:A_constrained}.

Here, we present an alternative proof of the convergence of constrained symmetric extensions in Theorem \ref{thm:constrained_dps_convergence} that has first been proven in Ref.~\cite{Berta2021} based on information-theoretic constructions. Our proof follows a different logic and directly makes use of the standard symmetric extension Theorem \ref{thm:dps_standard}. 
Furthermore, we make use of the following result that is well known in theory of symmetric tensors, see for instance \cite[Lemma 4.4]{katz1982conjecture} or \cite[Corollary 4.4]{Comon2008}.
\begin{lemma}[\cite{katz1982conjecture,Comon2008}] 
\label{lemma:linear_independence}
    Let $v_{1},\dots,v_{r} \in \mathbb{C}^{n}$ be pairwise linearly independent vectors. Then for every $k \geq r - 1$, the vectors $v_{1}^{\otimes k},\dots, v_{r}^{\otimes k} \in \mathbb{C}^{n^{k}}$ are linearly independent.
\end{lemma}

We will present an independent proof of Lemma~\ref{lemma:linear_independence}. This requires the following result that just states that if two vectors are linearly independent, then there is a functional that vanishes on one of them but not on the other. This is rather intuitive in the context of inner product spaces, nevertheless we will provide a formal proof. We will use $(\mathbb{C}^{n})^*$ to denote the dual vector space of $\mathbb{C}^{n}$.
\begin{lemma} \label{lemma:orthgonal_functional}
    Let $u, v \in \mathbb{C}^{n}$ be linearly independent vectors. Then there is a functional $\varphi \in (\mathbb{C}^{n})^*$ such that $\varphi(u) = 1$ and $\varphi(v) = 0$.
\end{lemma}
\begin{proof}
    Since $u \neq v$ then there is $\psi \in (\mathbb{C}^{n})^*$ such that $\psi(u) \neq \psi(v)$. We will now argue that we only need to address the case when $0 \neq \psi(u) \neq \psi(v) \neq 0$. If $\psi(v) = 0$, then $\varphi = \frac{\psi}{\psi(u)}$ is the functional existence of which we intend to prove. So we only need to treat the case when $\psi(v) \neq 0$.

    If $\psi(u) = 0$, then take any $\psi' \in (\mathbb{C}^{n})^*$ such that $\psi'(u) \neq 0$, such $\psi'$ must exist since $u \neq 0$, which itself follows from that $u$ and $v$ are linearly independent. We can now construct $\psi'' = \alpha \psi + \psi'$ where $\alpha \in \mathbb{C}$. We have $\psi''(u) = \psi'(u) \neq 0$, and we can choose $\alpha$ such that $\psi''(u) \neq \psi''(v)$, i.e., we can choose $\alpha$ such that $\psi'(u) \neq \alpha \psi(v) + \psi'(v)$, remember that we already argued that $\psi(v) \neq 0$. Finally, we can again only consider the case when $\psi''(v) \neq 0$ since if $\psi''(v) = 0$, then $\varphi = \frac{\psi''}{\psi''(u)}$ is the sought for functional.

    So let us consider the case when $0 \neq \psi(u) \neq \psi(v) \neq 0$. Denote $u' = \frac{\psi(v)}{\psi(u)} u$, then we have $\psi(u') = \psi(v)$, but since $u$ and $v$ are linearly independent we must have $u' \neq v$. But then there must exist a functional $\xi \in (\mathbb{C}^{n})^*$ such that $\xi(u') \neq \xi(v)$. Again if $\xi(v) = 0$, then we can construct the required functional $\varphi$ as above, so we only need to treat the case when $\xi(v) \neq 0$. We can now construct
    \begin{equation}
        \varphi' = \psi - \dfrac{\psi(v)}{\xi(v)} \xi.
    \end{equation}
    We have $\varphi'(v) = 0$ as requested. Moreover, we have
    \begin{equation}
        \varphi'(u) = \psi(u) - \dfrac{\psi(v)}{\xi(v)} \xi(u) = \psi(u) - \dfrac{\psi(u)}{\xi(v)} \xi(u') = \psi(u) \left( 1 - \dfrac{\xi(u')}{\xi(v)} \right).
    \end{equation}
    We already argued that $\psi(u) \neq 0$, moreover we have $\xi(u') \neq \xi(v)$ and so $1 - \frac{\xi(u')}{\xi(v)} \neq 0$. It follows that $\varphi'(u) \neq 0$. Thus, $\varphi = \frac{\varphi'}{\varphi'(u)}$ is the requested functional.
\end{proof}

\begin{proof}[Proof of Lemma~\ref{lemma:linear_independence}]
    Fix $k \geq r-1$ and $i \in \{1, \ldots, r\}$. Then for every $j \in \{1, \ldots, r\}$, $j \neq i$, there is a functional $\varphi_j \in (\mathbb{C}^{n})^*$ such that $\varphi_j(v_i) = 1$ and $\varphi_j(v_j) = 0$; the existence of such functional is due to Lemma~\ref{lemma:orthgonal_functional}. Moreover, let $\xi \in (\mathbb{C}^{n})^*$ be any functional such that $\xi(v_i) = 1$. Then we construct $\psi \in (\mathbb{C}^{n^{k}})^*$ as follows:
    \begin{equation}
        \psi = \left( \bigotimes_{j \in \{1, \ldots, r\}, j \neq i} \varphi_j \right) \otimes \xi^{\otimes (k - r + 1)}.
    \end{equation}
    Then we have $\psi(v_i^{\otimes k}) = 1$ and $\psi(v_j^{\otimes k}) = 0$ for $j \in \{1, \ldots, r\}$, $j \neq i$. Now assume there are $\alpha_1, \dots, \alpha_r \in \mathbb{C}$ such that $\sum_{n=1}^r \alpha_n v_n^{\otimes k} = 0$. Then we have
    \begin{equation}
        0 = \psi\left(\sum_{n=1}^r \alpha_n v_n^{\otimes k}\right) = \sum_{n=1}^r \alpha_n \psi(v_n^{\otimes k}) = \alpha_i.
    \end{equation}
    Using the same construction for every index, we get $\alpha_i = 0$ for every $i \in \{1, \ldots, r\}$ and thus the vectors are linearly independent.
\end{proof}

Now, we have prepared everything to prove Theorem \ref{thm:constrained_dps_convergence}.

\begin{proof}[Proof of Theorem~\ref{thm:constrained_dps_convergence}]
First, if $\rho_{\rm AB}$ is constrained separable, then it can be readily confirmed that it is $k$-constrained symmetric extendible.

Assume now that $\rho_{\rm AB}$ is $k$-constrained symmetric extendible for all $k\in \mathbb{N}$. This implies that $\rho_{\rm AB}$ is symmetric extendible in the usual sense that there exists a sequence of states $\rho_{\rm{A}^{k}\rm{B}}$ for $k \in \mathbb{N}$ such that $\rho_{\rm AB} = \Tr_{\rm{A}^{k-1}}(\rho_{\rm{A}^{k}\rm{B}})$. 
By Theorem \ref{thm:dps_standard}, we know that for every $k\in \mathbb{N}$ there exists a symmetric extension $\rho_{\rm{A}^{k}\rm{B}}$ of $\rho_{\rm AB}$ of the form
\begin{equation}
\label{eq:con_ext_proof}
    \rho_{\rm{A}^{k}\rm{B}} = \sum_{\lambda} p_{\lambda}^{(k)} \, {(\rho_{\rm A}^{\lambda})}^{\otimes k} \otimes \rho_{\rm B}^{\lambda}  
\end{equation}
which implies that $\rho_{\rm AB}$ can be written  as 
\begin{equation}
\label{eq:con_sep_proof}
    \rho_{\rm AB} = \sum_{\lambda} p_{\lambda}^{(k)} \, \rho_{\rm A}^{\lambda} \otimes \rho_{\rm B}^{\lambda}. 
\end{equation}
We now show that the existence of extensions of the form \eqref{eq:con_ext_proof} implies that $\rho_{\rm AB}$ is constrained separable, i.e., that $\Phi_{\rm A}(\rho_{\rm A}^{\lambda}) = a$ and $\Psi_{\rm B}(\rho_{\rm B}^{\lambda}) = b$ for all $\lambda$. First, by Carathéodory's theorem \cite{ECKHOFF1993}, it can without loss of generality be assumed that the sums in Eq.~\eqref{eq:con_ext_proof} and \eqref{eq:con_sep_proof} consist of finitely many summands and that the operators $\{\rho_{\rm A}^{\lambda}\}_{\lambda=1}^{n}$ with $n \in \mathbb{N}$ are pairwise linearly independent. Then, by Lemma \ref{lemma:linear_independence}, we have that $\{{(\rho_{\rm A}^{\lambda})}^{\otimes k}\}_{\lambda = 1}^{n}$ is linearly independent for $k \geq n-1$ which implies the existence of a set of operators $\{G_{\rm{A}^{k}}^{\lambda}\}_{\lambda = 1}^{n}$ such that $\Tr(G_{\rm{A}^{k}}^{\mu} {(\rho_{\rm{A}}^{\lambda})}^{\otimes k}) = \delta_{\mu \lambda}$. It is in the following assumed that $k\geq n$.

From the previous discussion, we have that 
\begin{align}
    p_{\mu}^{(k)} \, b &= \Tr_{\rm{A}^{k}}\left[(G_{\rm{A}^{k}}^{\mu} \otimes \mathds{1}_{\rm B}) \sum_{\lambda=1}^{n} p_{\lambda}^{(k)} \, {(\rho_{\rm A}^{\lambda})}^{\otimes k} \otimes b \right] \\
    &= \Tr_{\rm{A}^{k}}\left[(G_{\rm{A}^{k}}^{\mu} \otimes \mathds{1}_{\rm B}) \Tr_{\rm B}(\rho_{\rm{A}^{k}\rm{B}}) \otimes b\right]\\
    &\overset{\eqref{eq:B_constrained}}{=}\Tr_{\rm{A}^{k}}\left[(G_{\rm{A}^{k}}^{\mu} \otimes \mathds{1}_{\rm B}) (\textnormal{id}_{\rm{A}^{k}} \otimes \Psi_{\rm B}) \rho_{\rm{A}^{k}\rm{B}} \right]\\
    &= \Tr_{\rm{A}^{k}}\left[(G_{\rm{A}^{k}}^{\mu} \otimes \mathds{1}_{\rm{B}}) \sum_{\lambda=1}^{n} p_{\lambda}^{(k)} \, {(\rho_{\rm A}^{\lambda})}^{\otimes k} \otimes \Psi_{\rm B}(\rho_{\rm B}^{\lambda}) \right] \\
    &= p_{\mu} \, \Psi_{\rm B}(\rho_{\rm B}^{\mu})
\end{align}
for all $\mu = 1, \dots, n$ so that the operators $\rho_{\rm B}^{\lambda}$ fulfill the affine constraint. 

Let now $\{\widetilde{G}_{\rm{A}^{k-1}}^{\lambda}\}_{\lambda = 1}^{n}$ be a set of operators such that $\Tr(\widetilde{G}_{\rm{A}^{k-1}}^{\mu} {(\rho_{\rm A}^{\lambda})}^{\otimes k-1}) = \delta_{\mu \lambda}$. This leads to 
\begin{align}
    p_{\mu}^{(k)} \, a &= \Tr_{\rm{A}^{k-1}}\left[(\widetilde{G}_{\rm{A}^{k-1}}^{\mu} \otimes \mathds{1}_{\rm A}) \sum_{\lambda=1}^{n} p_{\lambda}^{(k)} \, {(\rho_{\rm A}^{\lambda})}^{\otimes (k-1)} \otimes a \right] \\
    &= \Tr_{\rm{A}^{k-1}}\left[(\widetilde{G}_{\rm{A}^{k-1}}^{\mu} \otimes \mathds{1}_{\rm A}) \Tr_{\rm AB}(\rho_{\rm{A}^{k}\rm{B}}) \otimes a\right]\\
    &\overset{\eqref{eq:A_constrained}}{=}\Tr_{\rm{A}^{k-1}}\left[(\widetilde{G}_{\rm{A}^{k-1}}^{\mu} \otimes \mathds{1}_{\rm A}) (\textnormal{id}_{\rm{A}^{k-1}} \otimes \Phi_{\rm A}) \Tr_{\rm B}(\rho_{\rm{A}^{k}\rm{B}}) \right]\\
    &= \Tr_{\rm{A}^{k-1}}\left[(\widetilde{G}_{\rm{A}^{k-1}}^{\mu} \otimes \mathds{1}_{\rm A})  \sum_{\lambda=1}^{n} p_{\lambda}^{(k)} \, {(\rho_{\rm A}^{\lambda})}^{\otimes k-1} \otimes \Phi_{\rm A}(\rho_{\rm A}^{\lambda}) \right] \\
    &= p_{\mu}^{(k)} \, \Phi_{\rm A}(\rho_{\rm A}^{\mu})
\end{align}
 for all $\mu = 1,\dots,n$ so that also the operators $\rho_{\rm B}^{\lambda}$ fulfil the affine constraints. This proves that $\rho_{\rm AB}$ is constrained separable. 
 \end{proof}

\subsection{Proof of Theorem \ref{prop:symmetric_ext_tester}}
\label{sec:mem_test_con_sep}
In this appendix, we prove the characterisation of memory-less testers by a hierarchy of semidefinite programs that is given in Theorem \ref{prop:symmetric_ext_tester} in the main text. 
\begin{theorem*}
    A single-copy tester $\{T_{\rm IO}^{i}\}_{i=1}^{N}$ is a convex combination of memory-less testers if and only if for all $k \in \mathbb{N}$ there exists a tester (not necessarily memory-less) $\{T_{\textnormal{Sym}(\rm{I},k) \rm{O}}^{i}\}_{i=1}^{N}$ such that $T_{\rm IO}^{i} = \Tr_{\rm{I}^{k-1}}(T_{\textnormal{Sym}(\rm{I},k) \rm{O}}^{i})$ where $\textnormal{Sym}(\rm{I},k)$ is the system that corresponds to the symmetric subspace of the Hilbert space of $k$ copies of the input system $\rm{I}$.
\end{theorem*} 

On the one hand, this characterisation of memory-less testers can be directly derived as a consequence of Theorem \ref{thm:constrained_dps_convergence}.
Here we show an alternative proof of the theorem, in which we make use of the following characterisation of convex combinations of memory-less testers. 

\begin{proposition}
\label{prop:sep_state_prep}
    A  single-copy tester $\{T_{\rm IO}^{i}\}_{i=1}^{N}$ is a convex combination of memory-less testers if and only if there exists a separable state $\rho_{\rm IE}$ and a measurement $\{M_{\rm EO}^{i}\}_{i=1}^{N}$ such that $T_{\rm IO}^{i} = \rho_{\rm IE} * M_{\rm EO}^{i}$.
\end{proposition}

\begin{proof}
    Let us first assume that the tester $\{T_{\rm IO}^{i}\}_{i=1}^{N}$ is a convex combination of memory-less testers, i.e.,
    \begin{equation}
        T_{\rm IO}^{i} = \sum_{\lambda} p_{\lambda} \; \Tilde{\rho}_{\rm I}^{\lambda} \otimes \Tilde{M}_{\rm O}^{i|\lambda}
    \end{equation}
    where $\Tilde{\rho}_{\rm I}^{\lambda}$ are states and $\{\Tilde{M}_{\rm O}^{i|\lambda}\}_{i=1}^{N}$ are measurements for all $\lambda$. Then we define a state $\rho_{\rm IE}$ and a measurement $M_{\rm EO}^{i}$ as 
    \begin{align}
        \rho_{\rm IE} &= \sum_{\lambda} p_{\lambda} \; \Tilde{\rho}_{\rm I}^{\lambda} \otimes \ketbra{\lambda}_{\rm E} \\
        M_{\rm EO}^{i} &= \sum_{\lambda} \ketbra{\lambda}_{\rm E} \otimes \Tilde{M}_{\rm O}^{i|\lambda}
    \end{align}
    In particular, notice that $\rho_{\rm IE}$ is a separable state. Now a straightforward calculation reveals that  $T_{\rm IO}^{i} =  \rho_{\rm IE} * M_{\rm EO}^{i}$ finishing the first part of the proposition.

    Let us now assume that there is a separable state $\rho_{\rm IE}$ and some measurement $M_{\rm EO}^{i}$ such that $T_{\rm IO}^{i} = \rho_{\rm IE} * M_{\rm EO}^{i}$. Then one calculates
    \begin{align}
        T_{\rm IO}^{i} &= \rho_{\rm IE} * M_{\rm EO}^{i} \\
        &= \sum_{\lambda} p_{\lambda} \; \Tr_{\rm E}\left[(\rho_{\rm I}^{\lambda} \otimes (\rho_{\rm E}^{\lambda})^{T} \otimes \mathds{1}_{\rm O})(\mathds{1}_{\rm I} \otimes M_{\rm EO}^{i})\right] \\
        &= \sum_{\lambda} p_{\lambda} \; \rho_{\rm I}^{\lambda} \otimes \Tr_{\rm E}[((\rho_{\rm E}^{\lambda})^{T} \otimes \mathds{1}_{\rm O}) M_{\rm EO}^{i}] \\
        &= \sum_{\lambda} p_{\lambda} \; \rho_{\rm I}^{\lambda} \otimes \Tilde{M}_{\rm O}^{i|\lambda}
    \end{align}
    where $\Tilde{M}_{\rm O}^{i|\lambda} := \Tr_{\rm E}[((\rho_{\rm E}^{\lambda})^{T} \otimes \mathds{1}_{\rm O}) M_{\rm EO}^{i}]$ are POVM elements on the system $\rm{O}$ for every $\lambda$ which implies that $\{T_{\rm IO}^{i}\}_{i=1}^{N}$ is a convex combination of memory-less testers.
\end{proof} 

\begin{proof}[Proof of Theorem \ref{prop:symmetric_ext_tester}]
    By Proposition \ref{prop:sep_state_prep}, convex combinations of memory-less testers are exactly given by $T_{\rm IO}^{i} = \rho_{\rm IE} * M_{\rm EO}^{i}$ where $\rho_{\rm IE}$ is a separable state. By Theorem \ref{thm:dps_standard}, $\rho_{\rm IE}$ is a separable if and only if there exists symmetric extensions $\rho_{\textnormal{Sym}(\rm{I},k)\rm{E}}$ for all $k\in N$ of $\rho_{\rm IE}$. Hence, one has 
\begin{align}
    T_{\rm IO}^{i} &= \rho_{\rm IE} * M_{\rm EO}^{i} \\
    &= \Tr_{\rm{I}^{k-1}}(\rho_{\textnormal{Sym}(\rm{I},k)\rm{E}}) * M_{\rm EO}^{i} \\
    & = \Tr_{\rm{I}^{k-1}}(\rho_{\textnormal{Sym}(\rm{I},k)\rm{E}} * M_{\rm EO}^{i}) \\
    &= \Tr_{\rm{I}^{k-1}}(T_{\textnormal{Sym}(\rm{I},k) \rm{O}}^{i})
\end{align}
where $\{T_{\textnormal{Sym}(\rm{I},k) \rm{O}}^{i}\}_{i=1}^{N}$ is a general single-copy tester with input system $\textnormal{Sym}(\rm{I},k)$ and output system $\rm{O}$ which completes the proof.
\end{proof}

\section{Proof of Theorem \ref{prop:shijun}}
\label{app:square_root_impossible}
\begin{proposition*}
    It is not possible to perfectly discriminate the ensemble given by the unitary operators $\{\mathds{1}, \sqrt{\sigma_{x}}, \sqrt{\sigma_{y}}, \sqrt{\sigma_{z}}\}$, explicitly defined by 
    \begin{equation*}
\mathds{1} =
\begin{pmatrix} 1&0\\0&1 \end{pmatrix}
\qquad
\sqrt{\sigma_x} =
\frac{e^{i\pi/4}}{\sqrt{2}}\begin{pmatrix} 1&i\\i&1 \end{pmatrix}
\qquad
\sqrt{\sigma_y} = \frac{e^{i\pi/4}}{\sqrt{2}}\begin{pmatrix} 1&-1\\1&1 \end{pmatrix}
\qquad
\sqrt{\sigma_z} = \begin{pmatrix} 1&0\\0&i \end{pmatrix}
\end{equation*}
using two-copy parallel testers. 
\end{proposition*}
\begin{proof}
Suppose, the unitary operators $\{\mathds{1}, \sqrt{\sigma_{x}}, \sqrt{\sigma_{y}}, \sqrt{\sigma_{z}}\}$ can be perfectly discriminated  using two-copy parallel testers. Then, since the environment system $\rm{E}$ is unlimited, every input state can be purified and there must exist a pure state $\ket{\psi} = \sum_{l=0}^{d_{\rm E}-1}\sum_{j,k=0}^{1} \alpha_{l,j,k} \ket{l,j,k}$ such that the output states $\ket{\psi_{a}} := (\mathds{1}_{\rm E} \otimes \sqrt{\sigma_{a}} \otimes \sqrt{\sigma_{a}})\ket{\psi}$ are pairwise orthogonal.

We now show that the collection of orthogonality conditions on the states $\ket{\psi_{a}}$ contradict the existence of $\ket{\psi}$.
    
We start with $\ket{\psi}$ and $\ket{\psi_{z}}$,
    \begin{align}
 \bra{\psi} \ket{\psi_{z}} 
    &= \sum \alpha_{l',j',k'}^\ast \; \alpha_{l,j,k} \braket{l'}{l} \mel{j'}{\sqrt{\sigma_{z}}}{j} \mel{k'}{\sqrt{\sigma_{z}}}{k}  \\
    &= \sum_{l} \bigl( \alpha_{l,0,0}^\ast \, \alpha_{l,0,0} - \alpha_{l,1,1}^\ast \, \alpha_{l,1,1}\bigr) + i \bigl( \alpha_{l,0,1}^\ast \, \alpha_{l,0,1} + \alpha_{l,1,0}^\ast \, \alpha_{l,1,0}\bigr) \\
    &= \left[ \sum_{l}  \bigl(\abs{\alpha_{l,0,0}}^2 - \abs{\alpha_{l,1,1}}^2 \bigr) \right] + i \left[ \sum_{l} \bigl( \abs{\alpha_{l,0,1}}^2 + \abs{\alpha_{l,1,0}}^2 \bigr) \right] \\
    &= 0
    \end{align}
Thus, one has $\alpha_{l,1,0} = \alpha_{l,0,1} = 0$ for all $l = 0,\dots,d_{\rm E}-1$. Hence, the state reduces to $\ket{\psi} = \sum_l \alpha_{l,0,0}\ket{l,0,0}+ \alpha_{l,1,1}\ket{l,1,1}$.

Next, we consider the orthogonality of $\ket{\psi}$ and $\ket{\psi_{x}}$
\begin{align}
 \bra{\psi} \ket{\psi_{x}}
   &=  \sum \alpha_{l',k',k'}^\ast \; \alpha_{l,k,k} \braket{l'}{l} \mel{k'}{\sqrt{\sigma_{x}}}{k} \mel{k'}{\sqrt{\sigma_{x}}}{k}  \\
    &= \sum \alpha_{l,k',k'}^\ast \; \alpha_{l,k,k} \mel{k'}{\sqrt{\sigma_{x}}}{k}^2 \\
    &= \frac{i}{2} \sum_{l}  \bigl(\abs{\alpha_{l,0,0}}^2 + \abs{\alpha_{l,1,1}}^2 - \alpha_{l,0,0}\alpha_{l,1,1}^\ast - \alpha_{l,1,1}\alpha_{l,0,0}^\ast)  \\
    &= \frac{i}{2} \sum_{l} \abs{\alpha_{l,0,0} - \alpha_{l,1,1}}^2 \\
    &= 0
\end{align}
so that $\alpha_{l,0,0} = \alpha_{l,1,1} =: \alpha_{l}$ for all $l$.

Finally, for $\ket{\psi}$ and $\ket{\psi_{y}}$,
\begin{align}
\bra{\psi} \ket{\psi_{y}}
   &=  \sum \alpha_{l'}^\ast \; \alpha_{l} \braket{l'}{l} \mel{k'}{\sqrt{\sigma_{y}}}{k} \mel{k'}{\sqrt{\sigma_{y}}}{k}  \\
    &= \sum \abs{\alpha_{l}}^2 \mel{k'}{\sqrt{\sigma_{y}}}{k}^2 \\
    &= 2i \sum_{l} \abs{\alpha_{l}}^{2} \\
    &= 0
\end{align}
It follows that $\alpha_{l} = 0$ and $\ket{\psi} = 0$ which is a contradiction. 
\end{proof}

\section{Proof of Theorem \ref{prop:parallel_vs_cadpt}}
\label{sec:prop_cadpt}
\begin{Theorem*}
    Every classically adaptive tester and every parallel tester is an adaptive tester.
    There exist testers which are parallel and not classically adaptive and vice versa, as well as testers which are parallel and classically adaptive. 
\end{Theorem*}
\begin{proof}
 We first show that every classically adaptive tester is an adaptive tester. Let the two-copy tester $\{T_{\rm \bf{IO}}^{i}\}_{i=1}^{N}$ be given by
\begin{align}
    \label{eq:ca_decomposition}
   T_{\rm \bf{IO}}^{i} &= \sum_{j=1}^{L} R_{\rm I_{1}O_{1}}^{j} \otimes S_{\rm I_{2}O_{2}}^{i|j}
\end{align}
where $\{R_{\rm I_{1}O_{1}}^{j}\}_{j=1}^{L}$ is a single-copy tester and $\{S_{\rm I_{2}O_{2}}^{i|j}\}_{i=1}^{N}$ are single-copy testers for all $j \in \{1,\dots,L\}$. We show that $\{T_{\rm \bf{IO}}^{i}\}_{i=1}^{N}$ satisfies the definition of an adaptive tester. From the definition of single-copy testers, one computes
\begin{align}
    W_{\rm \bf{IO}} &= \sum_{i=1}^{N} T_{\rm \bf{IO}}^{i} \\
    &= \sum_{j=1}^{L} R_{\rm I_{1}O_{1}}^{j} \otimes \sum_{i=1}^{N} S_{\rm I_{2}O_{2}}^{i|j} \\
    &= \sum_{j=1}^{L} R_{\rm I_{1}O_{1}}^{j} \otimes \sigma_{\rm I_{2}}^{j} \otimes \mathds{1}_{\rm O_{2}} \\
    &= G_{\rm \textbf{I}O_{1}} \otimes \mathds{1}_{\rm O_{2}}
\end{align}
with $G_{\rm \textbf{I}O_{1}} := \sum_{j=1}^{L} R_{\rm I_{1}O_{1}}^{j} \otimes \sigma_{\rm I_{2}}^{j}$ and 
\begin{align}
    \Tr_{\rm I_{2}}(G_{\rm \textbf{I}O_{1}})&=  \sum_{j=1}^{L} R_{\rm I_{1}O_{1}}^{j}\\
    &= \rho_{\rm I_{1}} \otimes \mathds{1}_{\rm O_{1}}
\end{align}
which together with Proposition \ref{prop:adpative_charcterisation} directly implies that $\{T_{\rm \bf{IO}}^{i}\}_{i=1}^{N}$ is adaptive.
That every parallel tester is an adaptive tester can also directly be seen from the definitions.

To see that the sets of classically adaptive and parallel testers have a non-empty intersection, we consider a special subset of classically adaptive testers. We assume that the testers $\{S_{\rm I_{2}O_{2}}^{i|j}\}_{i=1}^{N}$ for each $j\in \{1,\dots,L\}$ from the decomposition \eqref{eq:ca_decomposition} have the property that $\sum_{i=1}^{N} S_{\rm I_{2}O_{2}}^{i|j} = \sigma_{\rm I_{2}} \otimes \mathds{1}_{\rm O_{2}}$ with a state $\sigma_{\rm I_{2}}$ independent of $j$. Now, one calculates
\begin{align}
    \sum_{i=1}^{N} T_{\rm \bf{IO}}^{i} &= \sum_{j=1}^{L} R_{\rm I_{1}O_{1}}^{j} \otimes \sigma_{\rm I_{2}} \otimes \mathds{1}_{\rm O_{2}} \\
    &= \rho_{\rm I_{1}} \otimes \sigma_{\rm I_{2}} \otimes \mathds{1}_{\rm O_{1}O_{2}}
\end{align}
so that $\{T_{\textbf{I}\textbf{O}}^{i}\}_{i=1}^{N}$ satisfies the definition of parallel testers.

It remains to be shown that there are parallel testers which are not classically adaptive and vice versa. To this end, we provide an example for both cases.

Let first $\{\ket{\phi^{i}}\}_{i=1}^{4}$ be the usual two-qubit Bell states and consider the parallel tester $\{T_{\textbf{I}\textbf{O}}^{i}\}_{i=1}^{N}$ given by 
\begin{equation}
    T_{\rm \textbf{I}\textbf{O}}^{i} = \ketbra{\phi^{1}}_{\rm \textbf{I}} \otimes \ketbra{\phi^{i}}_{\rm \textbf{O}}.
\end{equation}
Notice that each tester element $T_{\rm \textbf{I}\textbf{O}}^{i}$ is entangled in the bipartition $\rm{I_{1}O_{1}|I_{2}O_{2}}$ which shows that the tester can not be classically adaptive as these are necessarily separable in that bipartition. 

Next consider the tester $\{T_{\textbf{I}\textbf{O}}^{i}\}_{i=1}^{N}$ given by 
\begin{equation}
    T^{i}_{\rm \textbf{I}\textbf{O}} = \sum_{j=0}^{1} \ketbra{0}_{\rm I_{1}} \otimes \ketbra{j}_{\rm O_{1}} \otimes \ketbra{j}_{\rm I_{2}} \otimes \ketbra{i}_{\rm O_{2}}.
\end{equation}
It can be readily confirmed that 
\begin{equation}
    \sum_{i=0}^{1} T^{i}_{\rm \textbf{I}\textbf{O}} \neq \Tr_{\textbf{\rm O}}\left(\sum_{i=0}^{1} T^{i}_{\rm \textbf{I}\textbf{O}}\right) \otimes \frac{\mathds{1}_{\rm \textbf{O}}}{d_{\rm \textbf{O}}}
\end{equation}
so that the tester is not parallel. 
\end{proof} 

\bibliographystyle{quantum}
\bibliography{ref}
\end{document}